\documentclass[12pt]{article}
\usepackage{amsmath,amsfonts}
\usepackage{natbib}
\usepackage{bbm}
\usepackage[plain]{algorithm2e}
\usepackage{graphicx}
\usepackage{rotating}
\usepackage{mathrsfs}
\usepackage{amsthm}
\usepackage{xcolor} 
\usepackage{subfigure}

\newcommand{\bx}[0]{\mathbf x}
\newcommand{\bh}[0]{\mathbf h}
\newcommand{\by}[0]{\mathbf y}
\newcommand{\Bf}[0]{\mathbf w}
\newcommand{\bbeta}[0]{\boldsymbol{ \beta}}
\newcommand{\bdelta}[0]{\boldsymbol{ \delta}}

\newcommand{\ip}[0]{{i^\prime}}
\newcommand{\jp}[0]{{j^\prime}}

\newtheorem{theorem}{\indent\sc Theorem}
\newtheorem{lemma}{\indent\sc Lemma}
\newtheorem{proposition}{\indent\sc Proposition}

\newtheorem{remark}{\indent\sc Remark}

\newcommand{\blind}{1}

\begin{document}
\def\spacingset#1{\renewcommand{\baselinestretch}%
{#1}\small\normalsize} \spacingset{1}

\title{A Bayesian Survival Tree Partition Model Using Latent Gaussian Processes}

\if1\blind
{
  \author{Richard D.\ Payne\footnote{Equal contribution to this work.} , Nilabja Guha\footnotemark[\value{footnote}] , Bani K.\ Mallick}
} \fi

\if0\blind
{
  \bigskip
  \bigskip
  \bigskip
  \medskip
} \fi

\date{}
\maketitle

\begin{abstract}
    Survival models are used to analyze time-to-event data in a variety of disciplines.  Proportional hazard models provide interpretable parameter estimates, but proportional hazards assumptions are not always appropriate.  Non-parametric models are more flexible but often lack a clear inferential framework.  We propose a Bayesian tree partition model which is both flexible and inferential.  Inference is obtained through the posterior tree structure and flexibility is preserved by modeling the the hazard function in each partition using a latent exponentiated Gaussian process.  An efficient reversible jump Markov chain Monte Carlo algorithm is accomplished by marginalizing the parameters in each partition element via a Laplace approximation.  Consistency properties for the estimator are established. The method can be used to help determine subgroups as well as prognostic and/or predictive biomarkers in time-to-event data.  The method is applied to a liver survival dataset and is compared with some existing methods on simulated data.
\end{abstract}

\noindent
{\it Keywords:}  Laplace approximation, reversible jump MCMC, time-to-event
\vfill

\spacingset{1.9} 

\section{Introduction}
Survival analysis for time-to-event data is used in a variety of disciplines, particularly in medical research. The Cox proportional hazard, accelerated failure time, and proportional odds models and their extensions are popular modeling tools to relate survival times with predictors or covariates in the presence of censorship. Though these are popular models, they are constrained by parametric assumptions such as proportional hazards, constant odds, linearity, etc. Survival trees and forests are popular non-parametric alternatives to these semi-parametric models which have the ability to identify non-linear interactions. One of the main benefits of tree models is their easy interpretation. Moreover, they can cluster subjects automatically according to their survival behavior based on covariates to identify prognostic groups. 

Tree-based survival analysis was introduced in the 1980s \citep{leblanc1993survival,bou2011review} primarily through the extension of classification and regression tree (CART) models \citep{breiman1984classification} to censored survival data \citep{ciampi1981approach,gordon1985tree}. The major difference among these papers are the use of different splitting criteria such as log-rank statistics, likelihood ratio statistics, Wicoxon-Gehan statistics, node deviance measures, and different impurity measures. A crucial issue of these tree models is the selection of a single tree of suitable size using forward or backward algorithms. 

In contrast to ordinary learning approaches which  construct one learner from training data, ensemble methods construct a set of learners known as base learners and combine them for prediction. Base learners are usually generated from the training data by a standard algorithm like decision tree, neural networks or other kinds of learning algorithms. There are two major kinds of ensemble methods based on the method of construction of the base learners. The first type is sequential ensemble methods where the base learners are generated sequentially and  the dependence among these base learners are used efficiently. The major example of this type  is the popular AdaBoost \citep{freund1997decision}. 
The alternative type of ensemble learning  constructs the base learners in a parallel way with bagging \citep{breiman1996bagging}  as an example. 

Recently, ensemble trees have become a popular alternative as a predictive survival model building tool \citep{ishwaran2004relative,hothorn2005survival,hothorn2006unbiased,ishwaran2008random,zhou2015rotation} using the concepts of bagging and random forests \citep{breiman1996bagging,breiman2001random}. In terms of application, \citet{garg2011phase} develop an approach to cluster hospital stays using survival trees and \citet{loh2015regression} use regression trees to identify subgroups with differential treatment effects. Indeed, there are a variety of other tree models which have been developed for use in survival and reliability models. \citet{levine2014bayesian} present some methodology for modeling tree survival functions with correlated observations. \citet{barrett2013gaussian} use an empirical Bayes approach by using Gaussian process regression on transformed survival times.  \citet{martino2011approximate,joensuu2012risk} and \citet{fernandez2016gaussian} use Gaussian processes in the survival model framework, and \citet{negassa2005tree}, \citet{su2008interaction} apply tree models for subgroup identification, but all retain the proportional hazards assumption.

Bayesian analysis of survival data is popular due to its flexibility to include prior information, suitability to develop complex models, and ability to construct efficient algorithms for full scale uncertainty quantification \citep{ibrahim2014b,sinha1997semiparametric}.   The Bayesian  classification and regression tree methods were developed by \citet{chipman1998bayesian,denison1998bayesian}. 
\citet{clarke2008bayesian} use empirical Bayes priors and ensembles of Weibull tree models to model survival distributions. 
The majority of Bayesian tree-based models in survival analysis are through ensembles of trees techniques (in a boosting way) using Bayesian additive regression trees (BART) \citep{chipman2010bart}. \citet{bonato2010bayesian} utilize BART to estimate latent parameters in a gamma process \citep{kalbfleisch1978non}, Weibull, and accelerated failure time models; \citet{henderson2017individualized} use BART in conjunction with a mean-constrained Dirichlet process mixture model to analyze heterogeneous treatment effects, which is reminiscent of \citet{kuo1997bayesian}; \citet{sparapani2016nonparametric} model survival times using a BART probit regression approach.

Even though machine learning algorithms usually focus on the predictive performance, the research in the field of interpretable machine learning \citep{molnar2020interpretable,doshi2017towards} has flourished in recent years due to its direct connection with artificial intelligence (AI) \citep{kuziemsky2019role}. For instance, AI is transforming the industry of healthcare rapidly and telehealth is one of its newest products. Telehealth utilizes digital information and communication technologies to access and manage healthcare services remotely. AI is booming in this industry of telemedicine where an efficient system for doctors is created to analyze, screen, and diagnose different disease conditions. However, a black-box machine-learning based predictive tool is not sufficient  and a more interpretable decision system is needed for this purpose. Hence, most of the ensemble methods (based on boosting or bagging logics) are not suitable in this framework.
 
To develop interpretable as well as flexible machine learning models for survival data,  we resort to a single tree model. One purpose of the proposed methodology is to provide a Bayesian tree framework which is simple to interpret (i.e.\ obtaining one tree) and flexible, i.e., moves beyond the proportional hazards model \citep{cox1972regression}. The  hazard function at each terminal node of the tree is modeled flexibly using an exponentiated Gaussian process. The posterior space of trees is efficiently searched via Markov chain Monte Carlo (MCMC) by obtaining the marginal likelihood of the observed data via a Laplace approximation \citep{tierney1986accurate}. The combination of interpretability coupled with flexibility makes this model well-suited for many problems, including subgroup identification and identifying prognostic and predictive biomarkers in medical research. In this Bayesian tree framework we sample from the posterior distribution of trees which provides a sample of multiple trees for decision making. 

Most existing Bayesian tree survival models have transformed the survival problem into a regression or classification framework. For example,  \citet{bonato2010bayesian} use latent variables and parametric models. Tree ensembles are used primarily on the covariate structure in hierarchical specifications through a regression structure. \citet{sparapani2016nonparametric} transform survival responses into binary outcomes and use BART in this classification setup.  In contrast, we model the survival responses directly using  different hazard functions at each terminal node, and  construct the partition of the covariate space using a single tree structure similar to that of the original Bayesian classification and regression tree (CART) papers \citep{chipman1998bayesian,denison1998bayesian}.   However, our proposed model is different than existing CART models where simple parametric models (like Gaussian or multinomial models) with finite numbers of parameters and conjugate priors have been used at the terminal nodes of the tree. In contrast, we allow an unknown hazard function to model the  survival data within each terminal node of the tree.  An exponentiated Gaussian process is used as a prior for this hazard function which is  flexible but creates a non-conjugate posterior distribution. That way, explicit marginalization is not possible to obtain the marginal likelihood function which is the key step of the existing Bayesian Tree search algorithms for CART as well as BART  models. We propose a Laplace-approximation based method to develop a novel tree search algorithm for non-conjugate Bayesian tree models. In general, this new approach will aid the Bayesian tree models to be applied in more complex modeling scenarios. 

Moreover, an important contribution of this paper is to provide theoretical properties of these Bayesian tree models. We  establish  consistency properties for our estimator under the tree partition model. Posterior consistency for survival functions in the Bayesian non-parametric setup  has been   investigated by modeling the  hazard  functions by completely random measures. Some earlier work using Dirichlet process priors can be found in  \citet{ghosh1995consistency}, where weak consistency was shown.  Later, using neutral to right processes such as the Beta and Gamma processes, weak consistency was shown in \citet{kim2003bayesian}. \citet{de2009asymptotics} show similar consistency results  where the hazard rates are a  kernel mixture of completely random measures. Here, we provide a framework for modeling the log hazard by a Gaussian process over partitions, where we can select the smoothness parameter based on observations for each partition. Under flexible conditions we show weak consistency and under slightly stronger conditions we establish the strong or $L^1$ consistency.  Partition modeling has been used for survival functions
\citep{ishwaran2008random} for a discrete nominal variable where $L^1$ consistency for survival function over a compact set  was shown for Kaplan-Meier type estimators \citep{ishwaran2010consistency}. Later,   developments using non i.i.d assumptions can be found in  \citet{cui2017consistency}. In a  different context, theoretical results have been provided for BART mainly in a regression setup (\citet{linero2018,Rock2017,Rock2018}). However, we are the first one to provide consistency theory for a Bayesian tree survival model, while modeling the sudden changes in survival function over the covariate space.



The paper is organized as follows:  Section 2 describes the model, Section 3 shows the theoretical properties of the model, Section 4 discusses the MCMC algorithm, Section 5 shows some simulated examples and an application to a liver survival dataset, Section 6 concludes.
%
%
\section{Survival models}
Let $t$ be a continuous nonnegative random variable representing the survival time of the subjects with the density function $f(t)$, the survival function $S(t)$ and the hazard function $h(t)= f(t) / S(t)$. In general, the hazard function depends on the survival time and the vector of covariates $\bx=(x_{1},x_{2},\ldots,x_{p})$. The Cox proportional hazards model \citep{cox1972regression} allow time and covariates to affect the hazard function separately, $h(t \mid \bx)=h_{0}(t){\rm exp}(\bx\bbeta)$, where $h_{0}$ is the baseline hazard and $\bbeta$ is a vector of regression parameters. This model has a strong assumption that the ratio of hazards for two individuals depends on the difference between their linear predictors at any time. 

We extend this class of models by developing the Treed Hazards Model (THM) which allows a nonlinear effect of survival time and covariates in a flexible and adaptive way. In the THM model we split the covariate space of $\bx$ with tree partitions and a fit piecewise constant hazard model within each partition element. Let's denote the tree as $T$ with $M$ terminal nodes and the vector of  hazards within tree partition elements (terminal nodes) as $\bh=(h_{1}(\cdot),h_{2}(\cdot),\cdots,h_{M}(\cdot))$ where $h_{i}(\cdot)$ is the hazard
function at the $i$th terminal node.  This is in contrast to specifying a finite number of parameters at each terminal node as in existing Bayesian CART models. The THM model describes the conditional distribution of $t$ given $\bx$ in this tree structure. If $t$ lies in the region corresponding to the $i$th terminal node then $t \mid \bx$ has the distribution $f(t\mid h_{i}(\cdot))$. In the next section we develop the hierarchical tree model for the survival data.

\subsection{The Bayesian hierarchical treed hazards model }
Let $\by = \{y_1,\ldots,y_n\}$ be the survival time for $n$ subjects (realizations of $t$),  $\bdelta = \{\delta_1, \ldots, \delta_n\}$ be the corresponding censoring indicators  where $\delta_j$ is 1 if $y_j$ is observed and 0 if $y_j$ is right censored. Then, we can write the posterior distribution of the  model unknowns $T$ and $\bh$ as $p(T,\bh \mid \by,\bx,\bdelta)\propto p(\by, \bdelta \mid T,\bh)p(\bh \mid T)p(T \mid \bx)$.  In the following sections, we specify the likelihood function $p(\by, \bdelta \mid T,\bh)$ and the prior distributions $p(\bh \mid T)$, $p(T \mid \bx)$ which will be used to obtain the posterior distribution.

\subsubsection{Tree partition and the prior distribution $p(T \mid \bx)$}
First, we describe the tree partition $T$ on the covariate space and the related prior distribution. We assume that the data can be partitioned into $M$ partition elements $\by_1,\ldots,\by_M$ such that $\bigcup_{i=1}^M \by_i = \by$ and $\by_i \cap \by_{i^\prime} = \emptyset$ for all $i \neq i^\prime$.  We further assume that the data within the $i$th partition element are conditionally independent given the hazard function, $h_i(t)$, and assume that the $M$ hazard functions are mutually independent.  The tree partition is constructed through a series of recursive binary splits of the covariate space \citet{chipman1998bayesian}. The binary tree $T$ subdivides the predictor space marginally by choosing one of the $p$ covariates from $\bx$ and then defining a partition rule.

We adopt the same prior as \citet{chipman1998bayesian} which is constructed by splitting the tree nodes with a designated probability and assigning a rule each time a node is split.  Let $S_\ip$ equal 1 if the $\ip$th node is split, otherwise 0, $\ip=1,\ldots,2M-1$ where $M$ is the number of terminal nodes on a tree.  Let $R_\jp = (v_\jp,r_\jp)$, $\jp=1,\ldots,M - 1$ represent the rules of the $M -1$ internal nodes.  Here, $v_\jp \in \{1,\ldots,p\}$ indicates which variable is to be split, and $r_\jp$ is the splitting rule on the $v_\jp$th variable.  The partition rule for continuous variables is formed by choosing a threshold and then grouping observations based on whether or not their covariate value is less than or equal to, or greater than, the splitting threshold.  For categorical variables, the split rule is defined as a subset of the unique categories for that variable, and observations are grouped based on whether their covariate value is in the subset or not.

Let $T = (S_1,\ldots,S_{2M -1},R_1,\ldots,R_{M-1})$ be the collection of split indicators and rules for a particular tree. Then the prior for a tree, $T$, can be evaluated as
\begin{eqnarray} \label{eq:prior}
\pi(T) = \prod_{\ip=1}^{2M-1} \pi(S_\ip) \prod_{\jp=1}^{M-1}\pi(R_\jp).
\end{eqnarray}

Following \citet{chipman1998bayesian}, we define $\pi(S_\ip = 1) = 1 - Pr(S_\ip = 0) = \gamma(1 + d_i)^{-\theta}$ where $d_\ip$ represents the depth of the $\ip$th node (i.e.\ the root node has a depth of 0, child nodes of the root node have depth 1, etc.).  Here, $\gamma$ and $\theta$ are hyperparameters chosen by the modeller and influence the prior for the general size and shape of the tree. Here, $\pi(R_\jp) = 1/(m_\jp m^\star_\jp)$ where $m_\jp$ represents the number of variables on the corresponding node that are available for splitting and $m^\star_\jp$ is the number of splitting rules available on the selected variable for the node to which the rule belongs.  The splitting rules at a given node are selected from the observed data which reach that node from its ancestry.  

\subsubsection{The Likelihood function $p(\by, \bdelta \mid T,\bh)$}
 The partition model assumes that given a tree, $T$, the hazard functions at each terminal node are independent and the data within a terminal node, given the hazard function at the node, are independent and identically distributed.  It is important to note that covariates are used to construct the tree, $T$, but are not used in the hazard functions.  Consequently, the likelihood function, given a tree, can be expressed as a simple product of likelihoods
\begin{equation}
p(\by, \bdelta \mid T, \bh) = \prod_{i=1}^{M} p(\by_i, \bdelta_i \mid h_i(\cdot)) \nonumber
\end{equation}
where $\by_i,\ \bdelta_i$ represent the $n_i$ responses and censoring indicators which belong to the $i$th terminal node, $p(\by_i, \bdelta_i \mid h_i(\cdot))$ represents the likelihood of $\by_i,\ \bdelta_i$ given the hazard function in the $i$th terminal node and $M$ is the total number of terminal nodes.

We write the density of an uncensored observation at the $i$th terminal node as $p_i(t)=h_i(t)e^{-\int_0^th_i(s)ds}$ for hazard function $h_i(\cdot)$. Let $p_c(\cdot)$ be the density of the censoring distribution, which has the same support as the true  densities $p_{0,i}(\cdot),\ i = 1, \ldots, m^*$ corresponding to true partition of the covariate space.

We place a Gaussian process prior on the log-hazard rate function on $[0,T_n]$, where $T_n=\max\{y_i\}$, $\log h_{i}(\cdot)=\beta_i+w_i(\cdot)$ and $w_i(t)$ follows Gaussian process with covariance  function $\tau_i^2K_{l_i}(s,t)=\tau_i^2\exp\{-|s-t|/l_i\}$, and let $\pi(\tau_i^2,l_i)=\pi(\tau_i^2)\pi(l_i)$ be the prior on the hyperparameters, and $\pi(\beta_i)$ be the prior on $\beta_i$.

\subsubsection{The hazards model $p(\bh \mid T)$ and the approximate likelihood}

In order to run a tractable Markov chain, a couple of alterations to the above model are used in practice.  First, an approximate likelihood and prior are used  where the domain is decomposed into small-width bins. Also, we normalize the domain between $[0,1]$ and use a mean zero Gaussian process under the proposed covariance kernel. 

The approximation assumes the hazard functions are piecewise constant within each bin of the form $h_i(t) = \lambda_{ik}$ for $t \in (s_{k-1}, s_k]$ with $z_k$ denoting the midpoint of the interval.  In all, there are $K$ bins with endpoints $(s_{k-1},s_k]$ for $k=1,\ldots,K$ where $s_0 =0$ and $s_K = 1$.  For simplicity, we use the same grid for each hazard function and assume that the bins all have equal width, inducing a regular grid for $(z_1,\ldots,z_K)$.   The likelihood for data at a single terminal node under these assumptions can be expressed as 
\begin{eqnarray*}
p(\by_i, \bdelta_i \mid \lambda_{i1},\ldots,\lambda_{iK}) &=& \prod_{j=1}^{n_i} \lambda_{ik_{ij}}^{\delta_{ij}} \exp\left\{ -\left[(y_{ij} - s_{k_{ij} - 1}) \lambda_{ik_{ij}} + I(k_{ij} > 1)\sum_{g=1}^{k_{ij} - 1} (s_g - s_{g-1}) \lambda_{ig}\right]\right\} \\
&=& \exp\left\{ \sum_{j=1}^{n_i} \delta_{ij} \log(\lambda_{ik_j}) - \left[\sum_{j=1}^{n_i} (y_{ij} - s_{k_{ij}-1})\lambda_{k_{ij}} + \sum_{k=1}^K (s_k - s_{k-1}) \lambda_{ik} m_{ik} \right] \right\} \\
&=& \exp\left\{ \sum_{k=1}^K \left[ n_{ik} \log(\lambda_{ik}) -  \lambda_{ik}\left( \sum_{j: k_{ij} = k} (y_{ij} - s_{k - 1}) + m_{ik}(s_k - s_{k-1}) \right)\right] \right\}
\end{eqnarray*}
where $y_{ij}$ is the $j$th response from the $i$th terminal node, $k_{ij}$ denotes the interval $y_{ij}$ belongs to, $m_{ik}$ is the number of observations of $\by_i$ which are greater than $s_k$, $n_{ik}$ is the number of uncensored observations of $\by_i$ which fall in interval $k$, and $\delta_{ij}$ is the censoring indicator for $y_{ij}$.

We model the hazard function $h_i(t) = \lambda_{ik} =  \exp(\omega_i(z_k))$ for $t \in (s_{k-1},s_k]$ where $\omega_i(t)$ is a zero-mean Gaussian process with covariance function $C_i(t,t^\prime) = \tau_i^2 \exp\{-|t - t^\prime|/l_i \}$.  This covariance kernel is chosen primarily for its computational advantages; see the appendix.  Given a set of $K$ locations, $\Bf_i = (\omega_i(z_1),\ldots,\omega_i(z_K))^T$  follows a multivariate-normal distribution with mean 0 and covariance matrix $\Sigma_i$ with the $i^\prime,j^\prime$th element of $\Sigma_i$ being $C_i(z_{i^\prime},z_{j^\prime})$. The posterior distribution of $\Bf_i$, assuming the covariance function parameters are known, is 
\[
 p(\Bf_i \mid \by_i, \bdelta_i) \propto p(\by_i, \bdelta_i \mid \Bf_i) \pi(\Bf_i) = (2\pi)^{-K/2}|\Sigma_i|^{-.5} \exp\left\{g_0(\Bf_i) \right\}
\]
where 
\[
g_0(\Bf_i) = \sum_{k=1}^K \left[ n_{ik} w_{ik} -  \exp(w_{ik})\left( \sum_{j: k_{ij} = k} (y_{ij} - s_{k - 1}) + m_k(s_k - s_{k-1}) \right)\right] -\frac{1}{2} \Bf_i ^T \Sigma^{-1}_i  \Bf_i
\]
and $w_{ik}$ is the $k$th element of $\Bf_i$.  Note that $g_0(\Bf_i)$ implicitly depends on the data, bins, and covariance function parameters.  Due to the assumption of independence between partitions, the full joint posterior distribution of a tree, $T$, with $M$ terminal nodes and $\Bf_1,\ldots,\Bf_M$ can be expressed as 
\[
p(T,\Bf_1,\ldots,\Bf_M \mid \by, \bdelta) \propto \pi(T) \prod_{i=1}^M p(\by_i, \bdelta_i \mid \Bf_i) \pi(\Bf_i)
\]
where $\pi(T)$ is the prior on the tree as in \eqref{eq:prior}.

The goal in many partition models is to find a suitable partition structure.  In a Bayesian framework, this corresponds to searching the posterior of the tree, i.e.,
\begin{eqnarray}
p(T \mid \by) &=& \int \ldots \int p(T,\Bf_1,\ldots,\Bf_M \mid \by) d\Bf_1 \ldots d\Bf_M \nonumber \\
&\propto & \pi(T) \prod_{i=1}^M \int p(\by_i \mid \Bf_i) \pi(\Bf_i) d\Bf_i. \label{eq:marginal}
\end{eqnarray}
To do this, often a model is chosen in which all node-specific parameters can be integrated out analytically.  When the integrals have an analytic form, often an efficient reversible jump MCMC chain \citep{green1995reversible} can be constructed to search the posterior partition space. Unfortunately there is no analytical solution for the integrals in \eqref{eq:marginal} for the present model.  To overcome this issue, a Laplace approximation is used to approximate each integral. Laplace approximations have been used previously in various survival model settings. \citet{levine2014bayesian} utilize a Laplace approximation to search trees in a Bayesian frailty model and \citet{martino2011approximate} use integrated nested Laplace approximations via data augmentation to fit a Cox proportional hazards model to avoid MCMC. In our method, the Laplace approximation is similar to \citet{levine2014bayesian}, but we do not have random effects and we also optimize the covariance function hyperparameters via empirical Bayes.  \citet{payne2020conditional} utilize a Laplace approximation in a partition model framework for non-survival conditional density estimation using a Voronoi partition.

In usual form, the Laplace approximation approximates the integral with a normal distribution derived from a Taylor expansion of $\log[p(\by_i \mid \Bf_i) p(\Bf_i)]$.  The Laplace approximation has a closed form given by 
\begin{equation} \label{eq:laplace}
\int p(\by_i, \bdelta_i \mid \Bf_i) \pi(\Bf_i) d\Bf_i \approx |D_i + \Sigma_i^{-1}|^{\frac{1}{2}}|\Sigma_i|^{-\frac{1}{2}} \exp\left\{g_0(\hat{\Bf}_i)\right\}
\end{equation}
where $D_i$ is a diagonal matrix with $j^\prime$th element $\exp(\hat{w}_{ij^\prime})\left[\sum_{j : k_{ij} = j^\prime} (y_{ij} - s_{j^\prime-1}) + (s_{j^\prime} - s_{j^\prime-1})m_{j^\prime}\right]$ and $\hat{\Bf}_i = \arg\max_{\Bf_i} p(\by_i \mid \Bf_i) \pi(\Bf_i)$, which can be obtained using Newton's method.

Thus far we have assumed that the covariance function parameters are fixed and known.  In practice, we employ empirical Bayes to set these values in each terminal node.  Specifically, we choose $\tau_i$ and $l_i$ to maximize $p(\by_i \mid \Bf_i) \pi(\Bf_i \mid l_i, \tau_i) \pi(l_i) \pi(\tau_i)$ where $\pi(l_i)$, $\pi(\tau_i)$ are t-densities with one degree of freedom and $\pi(\tau_i) = t(\tau_i / 10; 1)$ where $t(\cdot; 1)$ is the density of a t-distribution with one degree of freedom.

At this point, we have a way to approximate $p(T \mid \by)$, up to a normalizing constant, for a single tree utilizing Laplace approximations.  We now need to determine a way to search the approximate posterior.  As is common in many tree models such as \citet{chipman1998bayesian} and \citet{gramacy}, we employ Markov chain Monte Carlo methods.  Due to the varying number of parameters in the prior of the tree structure, a reversible jump Markov chain Monte Carlo algorithm \citep{green1995reversible} will be outlined in Section 4.

\section{Consistency results}
In this section, we show the density estimation consistency  when the true data generating density follows a partition model induced by a tree. If the survival function for the true data generating density has a tree partition form, then first we show that as $n$, the number of observations goes to infinity, the posterior distribution concentrates on  any weak neighborhood of the true data generating density with probability one. Then under slightly stronger  conditions, we establish the $L^1$ consistency.

Let  $P_1^*,\dots,P_{m^*}^*$ be the true partition of the covariate space denoted as  $\mathscr{P}^*=\{P^*_1,\dots,P_{m^*}^*\}$, and  let $h_{0,i}(t)>0,\ t>0$ for $i=1,\dots,m^*$,  be the true hazard function for partition $P_i^*$ and the corresponding true uncensored data generating  density for the $i$th partition is given by  $p_{0,i}(t)=h_{0,i}(t)e^{-H_{0,i}(t)}$ where $H_{0,i}(t)=\int_0^th_{0,i}(s)ds$. As mentioned earlier, $p_c$ is the density of the censoring distribution. Let $F_c$ and and $F_{0,i}$ be corresponding cdfs and $1-F_{0,i}(t)=e^{-H_{0,i}(t)}$. For a generic partition $\mathscr{P}=\{P_1,\dots, P_{k'}\}$ of the covariate space and the partition densities $p_{i}(\cdot)$,  for observation $y,\ {\bf x}$, and the censoring indicator function $\delta$, the data generating density $p(\cdot)$  is given by
\begin{eqnarray}
p(y,\delta \mid {\bf x},\mathscr{P})=\sum_{i=1}^{k'}\left[p_{i}(y)\int_y^\infty p_c(s)ds{\bf I}_{\delta=1}+p_c(y)\int_y^\infty p_{i}(s)ds{\bf I}_{\delta=0}\right]{\bf I}_{{\bf x}\in P_i}.
\label{likelihood1}
\end{eqnarray}

We show that under a Gaussian process tree partition model, for any bounded continuous function $g_1(\cdot)$, the posterior probability of the weak neighborhood around true data generating density $p^*$, $U^{\epsilon}_{g_1}=\{p:|\int g_1p-\int g_1p^*|<\epsilon\} $, where $\epsilon>0$,  approaches one almost surely, as $n$ approaches infinity.

We observe  $(y_1,\delta_1),\ldots,(y_n,\delta_n)$ at observed covariates ${\bf x}_1,\ldots,{\bf x}_n$, where ${\bf x}_j$ are i.i.d. for $j=1,\dots,n$. Let  $\#(T)$ denote the number of partitions/terminal nodes for a tree $T$. Denoting the generic prior distribution by $\pi(\cdot)$, for $T_n=\max_i\{y_i\}$,  from Section 2 we have the following prior:
\begin{eqnarray}
P_1,\dots,P_{\#(T)}&\sim&\pi(T),\nonumber\\
\log h_i(t)=\beta_i+w_i(t) &\sim & GP(\beta_i,\mathscr{K}), \nonumber \\
\mathscr{K}(s,t)&=&\tau_i^2 K\left(\frac{s}{l_i},\frac{t}{l_i}\right),\ t\in [0,T_n];\ i=1,\cdots,\#(T),\nonumber\\
\pi(\tau^2_i,l_i)&=&\pi(\tau^2_i)\pi(l_i),\nonumber\\
\beta_i&\sim&\pi(\beta_i).
\label{prior}
\end{eqnarray}

We assume the following:
 \begin{itemize}
 \item[A1] $\log (h_{0,i}(t))=\int_0^\infty \tilde{g}_i(s)K(s/l_i^*,t/l_i^*)ds+\beta_{0,i}$, for some bounded, Lipschitz-continuous and integrable functions $\tilde{g}_i$'s, with $ w_{0.i}(t)=\int_0^\infty \tilde{g}_i(s)K(s/l_i^*,t/l_i^*)ds.$
 \item[A2] $E_{p_{0,i}}[t],\ E_{p_c}[t]<\infty$.
 \item[A3] $\max_t|\int_0^{T'}\tilde{g}_i(s)K(s/l_i^*,t/l_i^*)ds-w_{0,i}(t)|=o(1/T')$, as $T'\uparrow \infty$,  where $w_{0,i}(t)$'s are defined in [A1].
 \item[A4]  The covariates are assumed to be bounded. Each continuous covariate is partitioned into $m_{n}$ equi-spaced  grids for the tree splitting rule, such that $m_n\rightarrow \infty$ as $n\uparrow \infty$ and   $-\log \pi(T)= o(n)$ for any  tree $T$ such that $\#(T)<\infty$. 
 \item[A5] $t^\alpha\preceq max\{-\log (1-F_c(t)),-\log (1-F_{0,i}(t))\}; t>0$, for some $\alpha>0$ and for all $i=1,\ldots,m^*$. Also, $p_c$ and $p_{0,i}$'s have the same support.  Here $\preceq$ denotes less than equal to a constant multiple.
 \item [A6]Priors $\pi(\tau^2_i)$ and  $\pi(l_i)$ have strictly positive continuous densities on the positive real line, and $\pi(\beta_i)$ has a strictly positive  continuous density on the real line. 
 \end{itemize}
 
 Conditions [A1] and [A3] ensure that a constant shift of the  true log-hazard function remains in the Gaussian process support. Conditions [A1], [A4], [A5], and [A6] guarantee enough prior probability in a small neighborhood around  the true data generating density under the partition-tree model. Condition [A2] is needed for prior Kullback-Leibler support.
  
 Suppose the  true partition $\mathscr{P}^*$  is achieved for  a tree $T^*$, and let $\hat{\mathscr{P}}^*=\hat{\mathscr{P}}_T^*=\{\hat{P}_l^*:l=1,\cdots,m^*\}$ corresponding to $T=\hat{T}^*$ be an approximating partition, with $\hat{T}^*$ having the same number of terminal nodes as in $T^*$. We define the following distance between these two partitions  by
\[d(\mathscr{P}^*,\hat{\mathscr{P}}^*)=\text{min}_{c(1),\dots,c(m^*)} \sum_{i=1}^{m^*}\mu_x(P_i^*\Delta\hat{P}_{c(i)}^*)  \] where $c(1),\ldots,c(m^*)$ is a permutation of $1,\ldots,m^*$, and  $\mu_x$ is a measure associated with the covariate space which can be an absolutely continuous Lebesgue measure, a counting measure, or a product of both. Let $\mathscr{P}^*\cap \hat{\mathscr{P}}^*=\cup_{i=1}^{m^*}P_i^*\cap \hat{P}^*_{\hat{c}(i)}$ and  $\mathscr{P}^*\Delta \hat{\mathscr{P}}^*=\cup_{i=1}^{m^*}P_i^*\Delta \hat{P}^*_{\hat{c}(i)}$, where $\hat{c}(\cdot)$ is the minimizing permutation in the definition of $d(\mathscr{P}^*,\hat{\mathscr{P}}^*)$. For notational simplicity,  later $P_i^*\cap \hat{P}^*_{\hat{c}(i)}$ will be denoted by $P_i^*\cap \hat{P}^*_i$, that is $\hat{P}^*_i$ approximates $P^*_i$.

Now, we can state the following results. The proofs for the results are relegated to the appendix. 

 \begin{proposition}
Under [A5], $T_n=o(n^\beta)$ almost surely,  for any $\beta>0$, where $T_n=\max_i\{y_i\}$.
\label{max_sup}
\end{proposition}
 \begin{proposition}
For partition $\mathscr{P}^*$, given $\xi>0$ there exists a set $S_T^{\xi}=\{T: d( \hat{\mathscr{P}}_T^*,\mathscr{P}^*)<\xi , \#(T)=\#(T^*)\}$ such that  under [A4], $-\log \pi(S_T^{\xi})=o(n)$.
\label{tree_support1}
\end{proposition}
Next, using Proposition \ref{max_sup} and the results from \citet{van2008reproducing},  \citet{aurzada2008small}  on Gaussian process small ball probability, we show the following. 
\begin{lemma}

Let $\tilde{{w}}_{0,i}^n(t)=\int_0^{T_n}\tilde{g}_i(s)K(s/l_i^*,t/l_i^*)ds$ for $0\leq t\leq T_n$.  For any $\xi>0$, $S^i_\xi=\{w_i(\cdot):|w_i(t)-\tilde{w}^n_{0,i}(t)|<\xi/T_n, 0\leq t\leq T_n\}$. Then under  the Gaussian process prior on $w_i(\cdot)$ from \eqref{prior}, under A5, A6, $-\log\pi(S^i_\xi)=o(n)$.
\label{gp_support}
\end{lemma}
 
 Proposition \ref{tree_support1} and Lemma \ref{gp_support} establish a large enough prior probability around the truth, and Proposition \ref{max_sup} bounds $T_n$ and is used in Lemma \ref{gp_support}.  Now, we state the result determining the Kullback-Leibler neighborhood around the truth, which combined with Propositions \ref{max_sup}, \ref{tree_support1}, and Lemma \ref{gp_support} will guarantee that we have sufficient prior probability around a small Kullback-Leibler neighborhood around the truth. The support for the $p_{0,i},\ p_c$ are assumed to be on the positive real line, under which $T_n$ increases to infinity almost surely,  as $n$ increases to infinity. If the supports are the same for $p_{0,i},\ p_c$ and bounded then  the integral in assumption [A1] is replaced by the integral over the support of $p_{0,i}$'s and the  following Theorem \ref{thm1} will hold.

 \begin{lemma}
Let $p=\hat{p}$ be an estimator of $p^*(t,\delta,{\bf x})$ with hazard function $h_i(t)=e^{w_i(t)+\beta_i}$  for $0\leq t\leq T_n$ in partition $\hat{P}_i$ and $h_i(t)=h_i(T_n)$ for $t>T_n$, where $\hat{\mathscr{P}}^*=\{\hat{P_i}^*\}_{i=1,\dots,m^*}$ a partition of the covariate space,  and $d(\mathscr{P},\hat{\mathscr{P}^*})<\epsilon$.
Then, given $\epsilon>0$, there exists $\xi>0$ such that  for  $|\beta_i-\beta_{0,i}|<\xi$ and  $|w_i(t)-w_{0,i}(t)|<\xi$, for $0\leq t\leq T_n$, in $\hat{\mathscr{P}^*}\cap {\mathscr{P}^*}$, we have under [A1]--[A3] ,  $KL(p^*,p)=\int p^* \log\frac{p^*}{p}<k\epsilon$ for large $n$, for some universal constant $k>0$ .
\label{kl_sup}
\end{lemma}
 
 Finally, we can state the following consistency theorem.
 
 \begin{theorem}

For any bounded continuous function $g$, and for $\epsilon>0$, let $U^{\epsilon}_g=\{p:|\int g p^*(\cdot)-\int g p(\cdot)|<\epsilon\}$, be a weak neighborhood of the true joint density $p^*(\cdot)$. Under the model in equation \ref{likelihood1}, prior in equation \ref{prior}, under assumptions [A1]--[A6], and under the true data generating distribution $p^*$, $\Pi_n(U^{\epsilon}_g\mid \cdot)\rightarrow 1$ almost surely as $n\rightarrow \infty$, where $\Pi_n(\cdot)$ denotes the posterior distribution based on $n$ observations.
\label{thm1}
\end{theorem}

Theorem \ref{thm1} establishes consistency in  weak topology with $\Pi_n(\mid \cdot)$ being weakly consistent at $\{ {p_{0,i}}_{i=1,\cdots,m^*},p_c,\mathscr{P}^* \}$. Even though  we state  Theorem \ref{thm1} under the support condition given in [A3], the consistency result will hold true for more  general support  under additional mild conditions, given in the following.

\begin{remark}
Theorem \ref{thm1} holds for piecewise constant hazard functions which eventually converges to a constant that is  $\log h_{0,i}(t)=c_{0,i}$ for $t\geq T_{0,i}\geq 0$ for $i=1,\ldots,m^*$, if we assume $c_0l_i^{\alpha_{0,i}}\leq\pi(l_i)$ for some $\alpha_{0,i}>0,c_0>0$ for $|l_i|\leq c_l, c_l>0$. 
\label{rmrk1}

\end{remark}

\begin{remark}
Theorem \ref{thm1} holds for bounded Lipschitz continuous $\log h_{0,i}(t)$'s  if we assume $c_0l_i^{\alpha_{0,i}}\leq\pi(l_i)$ for some $\alpha_{0,i}>0,c_0>0$ for $|l_i|\leq c_l, c_l>0$.
\label{rmrk2}
\end{remark}
We can also slightly relax the tree support condition as follows. 
\begin{remark}

If the true partition is not induced by a tree, but given any $\epsilon>0$, can be approximated by a partition ${\mathscr{P}}_1^*$, induced by a  tree ${T}_1^*$ (may depend upon $\epsilon$), such that  $d(\mathscr{P}^*,{\mathscr{P}}_1^*)<\epsilon$, then Theorem \ref{thm1} holds.
\label{rmrk3}

\end{remark}

Let $D_n$ be the maximum number of terminal nodes allowed for $n$ observations.  Note that  $D_n$ will induce a total  less than  $D_n{(mm_n)}^{D_n}$ many possible tree formations where $m_n$ is the number of grids for a covariate and we have $m$ covariates and they are continuous. Establishing $L^1$ or strong consistency will require controlling or bounding the size of this model space. Let $\{M_n\}_{n\geq 1}$ be a  positive sequence increasing to infinity.
For establishing strong consistency, we will need to assume the following conditions. 
For  positive sequences $\{a_n\}_{n\geq 1},\{b_n\}_{n\geq 1}$,  $a_n\prec b_n$, implies $a_n/b_n\rightarrow 0$, as $n\uparrow \infty$.

\begin{itemize}

\item[B1]  The covariates are assumed to be bounded. Each continuous covariate is partitioned into $m_{n}$ equi-spaced  grids for the tree splitting rule, such that $m_n\rightarrow \infty$ as $n\uparrow \infty$, and $\log m_n\sim O(\log n)$, and   $-\log \pi(T)= o(M_n^2)$ for any  tree $T$ such that $\#(T)<\infty$.
\item[B2]  $D^{1+\delta'}_nM_n^2  \log m_n=o(n)$, for some $\delta'>0$. 
\item[B3] $M_n^2\prec n$ and $n/M_n^2= o(n^{\delta})$ for any $\delta>0$.
 
 \item[B4]  $\log (h_{0,i}(t))=\int_0^\infty \tilde{g}_i(s)K(s/l_i^*,t/l_i^*)ds+\beta_{0,i}$, for some bounded, Lipschitz-continuous, functions $\tilde{g}_i$'s, that are supported on a compact set. Let $ w_{0.i}(t)=\int_0^\infty \tilde{g}_i(s)K(s/l_i^*,t/l_i^*)ds.$

 \item[B5] $\pi(l_{i'})$ is supported on equi-spaced $M_n^{d'}$, $d'>0$, many grids in  a closed interval  not containing zero, that is a subset of the positive real line,   and that contains  $l_i^*$ as an interior point, for any $i',i$, and  $-\log \pi(l)=o(M_n^2)$, for $l$ in the support. 
 
 Also,
 $\pi(\tau^2_i)$  and $\pi(\beta_i)$ have strictly positive  continuous densities on their supports, positive real line and real line, respectively, and $min\{-\log P(\tau_i^2>t),-\log P(|\beta_i|>t)\}\succeq t^{\alpha'}$, $\alpha'>0,t>0$.

\end{itemize}

Condition $B2$ bounds the $L^1$ covering number of the prior space. Conditions $B3,B4$ are needed for   controlling  the ratio of the joint density  for the parameter value around the truth and the true joint density, almost surely, for a small neighborhood around the truth and to guarantee exponential decaying prior probability outside a compact set. Condition  $B1$ is needed to approximate the true data generating tree, and  establish sufficient prior probability in that small neighborhood. 

We can now state the following strong consistency result. For the strong consistency we consider estimated densities that are supported on $[0,T_n]$ and show that they will remain in a small $L^1$ neighborhood of the true density with probability one. 

\begin{theorem}

Let $\epsilon>0$ and $U^{\epsilon}=\{p:\int | p^*(\cdot) -p(\cdot)|<\epsilon\}$, be a $L^1$ neighborhood of the true joint density $p^*(\cdot)$. Under the model in equation \ref{likelihood1}, prior in equation \ref{prior}, under assumptions [A2],[A4]--[A5],[B1]--[B5] and under the true data generating distribution $p^*$, $\Pi_n(U^{\epsilon}\mid \cdot)\rightarrow 1$ almost surely as $n\rightarrow \infty$, where $\Pi_n(\cdot)$ denotes the posterior distribution based on $n$ observations.
\label{thm2}
\end{theorem}

Technical details and related proofs of the results are given in  Appendix ii.
\section{MCMC}
Since the dimension of the tree parameter space is variable, we use a reversible jump MCMC algorithm to search the posterior parameter space  \citep{green1995reversible}. The algorithm is inspired by those of \citet{chipman1998bayesian}  and \citet{gramacy} with a few modifications. The primary goal in constructing the algorithm is to obtain draws from $\tilde{p}(T \mid \by)$, the approximate posterior obtained using the Laplace approximations.  This is accomplished in the reversible jump MCMC algorithm by modifying an existing tree with one of four possible moves: grow, prune, change, and swap.

\begin{itemize}
\item {\bf Grow:} Randomly select a terminal node and add two child nodes based on a split rule.  The split rule is selected by first randomly selecting a variable and then randomly selecting an available rule of that variable (i.e.\ the rule is drawn from the prior).  This increases the number of terminal nodes of the tree by one, thereby ``growing" the tree.

\item {\bf Prune:} The opposite of a grow step.  Randomly select a node whose children are both terminal nodes and delete the split rule and the children nodes.  This reduces the number of terminal nodes of the tree by one, thereby ``pruning" the tree.

\item {\bf Change:} Randomly select an interior node and change the existing split rule.  If the current rule is based on a categorical variable, a new rule is drawn randomly from the prior by first randomly selecting a variable, and then randomly selecting a rule from that variable.  If the current rule is based on a numeric variable, there are two possible changes that can be made with pre-specified probability:  1) Move the current rule to the next largest or smallest rule available for the current variable, 2) Draw a new rule from the prior.

\item {\bf Swap:} Randomly choose a parent and child pair which are both internal nodes and swap their rules.  If, however, the parent and child's rules are based on the same continuous variable, a rotate step is employed \citep{gramacy}.  If both children of the parent node have the same rule, both children swap their rules with the parent node.
\end{itemize}


The MCMC algorithm proceeds as follows:

\begin{enumerate}
	\item Initialize a starting tree, $T$.
	\item Propose a new tree, $T^1$, from the current tree using a grow, prune, change, or swap step (as appropriate based on current tree).
	\item Accept $T^1$ with probability 
	\begin{eqnarray} \label{eq:acceptance}
	\alpha = \min\left\{1,\frac{q(T \mid T^1) \tilde{p}(T^1 \mid \by)}{q(T^1 \mid T) \tilde{p}(T \mid \by)}  \right\}
	\end{eqnarray}
	
	where $q(\cdot \mid \cdot)$ is the proposal distribution and $\tilde{p}(T \mid \by)$ is the approximate posterior of $T$ obtained from the Laplace approximation; otherwise keep the current tree, $T$. 
	\item Repeat steps 2-3 until the stationary distribution is reached and the desired number of samples is obtained.
\end{enumerate}

In practice we replace the approximate posterior of the tree in \eqref{eq:acceptance} with the likelihood-prior product in \eqref{eq:marginal} since the normalizing constants cancel.

From their inception, Bayesian tree models have been notorious for their difficulty in mixing when using MCMC procedures.   Various methods have been proposed to overcome this issue including restarting the MCMC chain multiple times \citep{chipman1998bayesian}, parallel tempering \citep{gramacy2010categorical}, and multiset sampling \citep{leman2009}.  We implement parallel tempering \citep{geyer1991markov} modeled after \citet{gramacy2010categorical}.  Additional technical details regarding the MCMC, are provided in the appendix.


\section{Simulations \& Applications}
\subsection{Preliminaries}
For the following analyses, we use the same set of tuning parameters for the MCMC chain.  We set the probability of grow, prune, change, and swap steps to be .25.  For the change step for a continuous variable, we choose to move the rule to an adjacent rule with probability .75, otherwise a rule is drawn from the prior.  We run each analysis with 8 parallel chains with the minimum inverse-temperature of .1, and propose a trade between the parallel chains every 10 MCMC iterations. Trees are required to have a minimum of 25 observations at a terminal node.
\subsection{Tree partition with non-proportional hazards} \label{sec:tp}
Our first simulation demonstrates the flexibility and interpretability of the THM.  One thousand data points were simulated from the following model:
%
\begin{equation} \nonumber 
\begin{aligned}
X_1  & \sim  \text{Uniform}(0,10), \\
X_2 & \sim  \text{Multinomial}(\text{A,B,C,D};\ .25,.25,.25,.25), \\
Y & \sim  \begin{cases}
\text{Weibull}(5,2) & \text{ if } X_1 > 5 \text{ and } X_2 \in \{\text{A,B} \} \\
\text{Weibull}(1,5) & \text{ if } X_1 \leq 5 \text{ and } X_2 \in \{\text{A} \} \\
\text{Weibull}(.5,.9) & \text{ if } X_1 \leq 5 \text{ and } X_2 \in \{\text{B} \} \\
\text{Weibull}(5,5) & \text{ if } X_1 \leq 3 \text{ and } X_2 \in \{\text{C,D} \} \\
\text{Weibull}(.5,.5) & \text{ if } 3 < X_1 \leq 7 \text{ and } X_2 \in \{\text{C,D} \} \\
G & \text{ if } X_1 > 7 \text{ and } X_2 \in \{\text{C,D} \} \\
\end{cases}
\end{aligned}
\end{equation}
%
where $G$ is a random variable drawn from an arbitrary piecewise linear survival function.  The data were also right censored with a 5\% censoring rate. The Markov chain Monte Carlo algorithm ran for $10^4$ iterations.  The tree with the highest posterior probability from the chain is shown in Figure \ref{fig:sim9tree} along with the posterior and true survival curves.  Note that since we are using observed data values to choose the rules at each node, the rules for numeric variables will not be exactly the same as those used to generate the tree.  This example highlights the strength of the THM, simultaneously producing easy interpretation and flexibility.  Interpretation is simple since only one tree is needed, but the the various shapes of the survival curves can still be modeled flexibly at each terminal node.


\begin{figure}
\begin{center}
  \subfigure[]{\includegraphics[trim=40 30 40 10, clip,width = .45\textwidth]{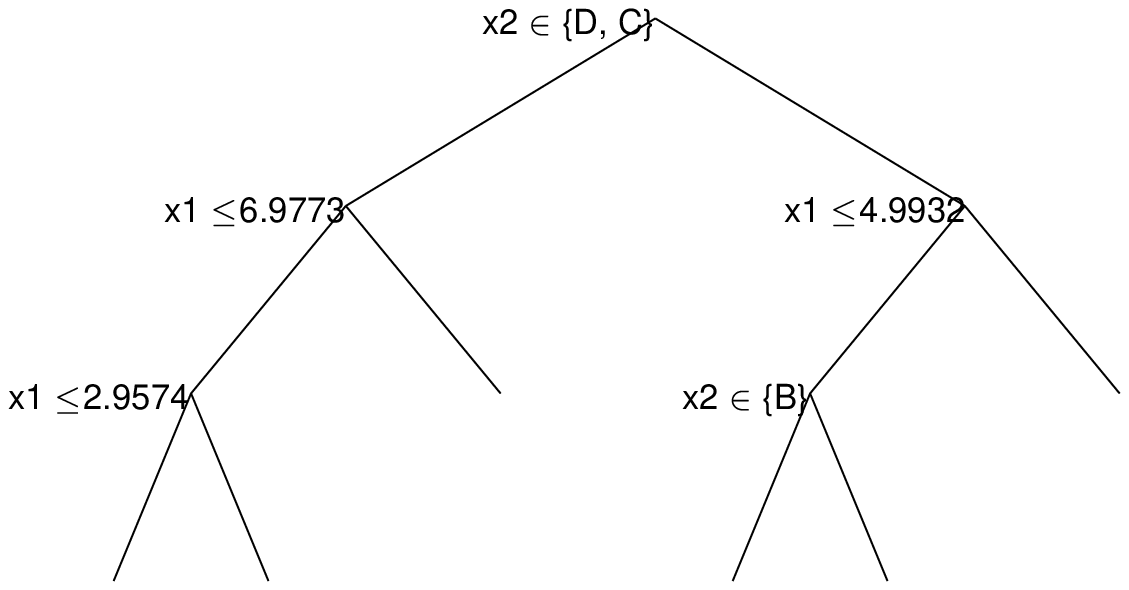}}
  \subfigure[]{\includegraphics[trim=0 10 0 10, clip,width = .45\textwidth]{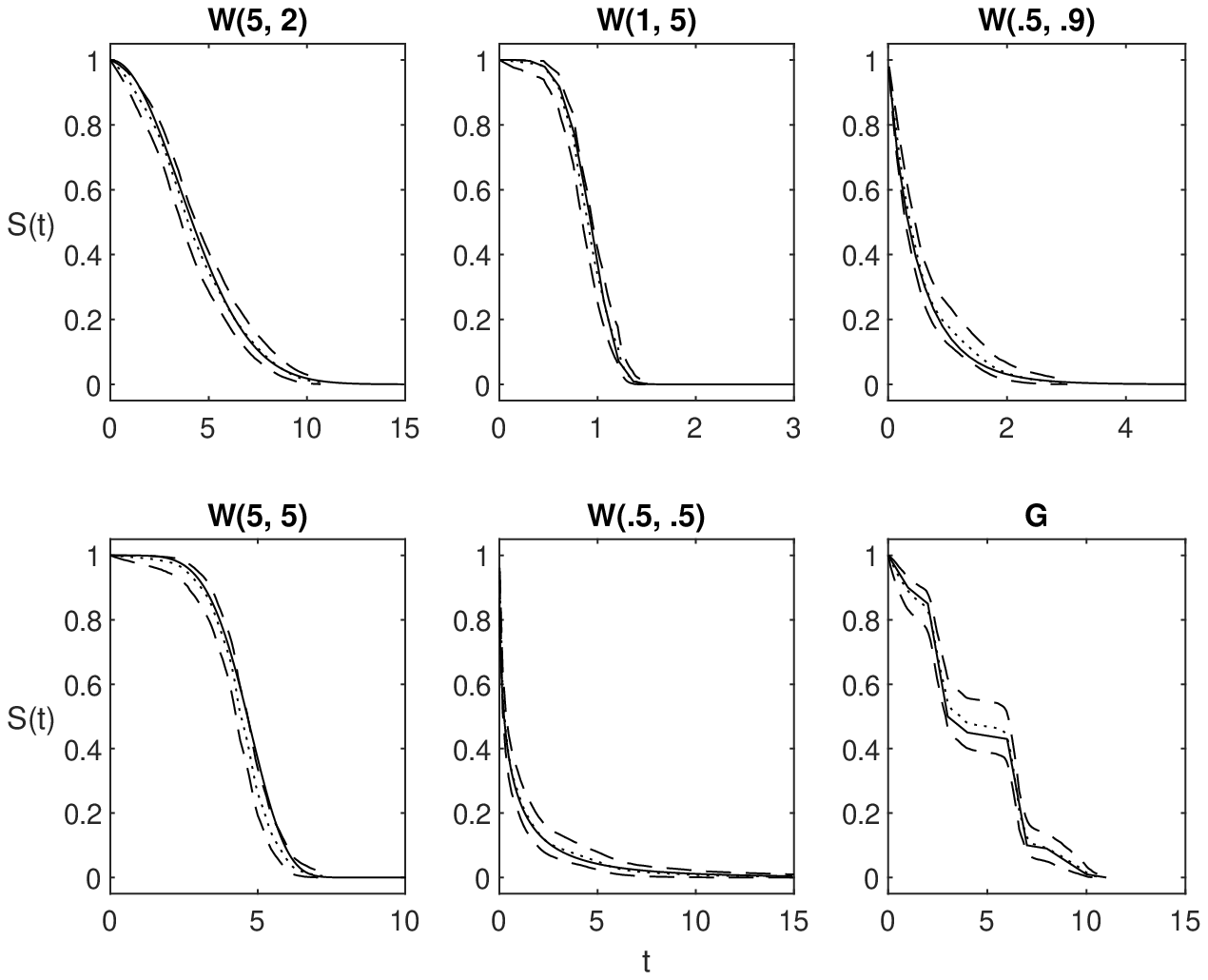}}
  \caption{(a) The tree with the highest posterior probability from the model in Section \ref{sec:tp}.  (b) The posterior mean (dotted), 95\% credible intervals (dashed), and the true survival curve (solid) in each partition element.}
\label{fig:sim9tree}
\end{center}
\end{figure}


\subsection{Cox proportional hazards} \label{sec:lph}
Although the present partition model assumes the data are generated from a tree structure, it still performs well when this assumption does not hold, e.g., when dealing with the classic linear proportional hazards model \citep{cox1972regression}.  One thousand observations were simulated from an exponential proportional hazards model with an 8\% censoring rate.  The model is
\begin{equation}
\begin{aligned}
h(t \mid \mathbf{X} = \mathbf{x}) =\ & \frac{1}{2} \exp\left\{\mathbf{x} \mathbf{\beta} \right\} \\
\mathbf{X} =\ & (X_1,\ldots,X_6) \\
X_1,\ldots,X_5  \sim\ & \text{Uniform(0,1)} \\
X_6  \sim\ & \text{Bernoulli(.5)} \\
\beta =\ & (-1, 1, 2, 0, 0, -2)^T.
\end{aligned}
\end{equation}
where $h(\cdot \mid \cdot)$ represents the hazard function. 

Note that we have two spurious covariates ($X_4,X_5$) which do not influence the hazard function.  The MCMC algorithm ran for $10^4$ iterations.  In the tree with the highest posterior probability (8 terminal nodes), neither of the two spurious covariates were included.  To give a measure of performance, Figure \ref{fig:sim7survival} plots the posterior survival curve and the exponential survival curve using the median true hazard value of the data that reach the terminal node.  This comparison serves as an indicator of where the survival distribution should lie in that region of the covariate space, and indicates the model captures the important features of the curves well.


\begin{figure}[t]
\begin{center}
\includegraphics[scale=.55]{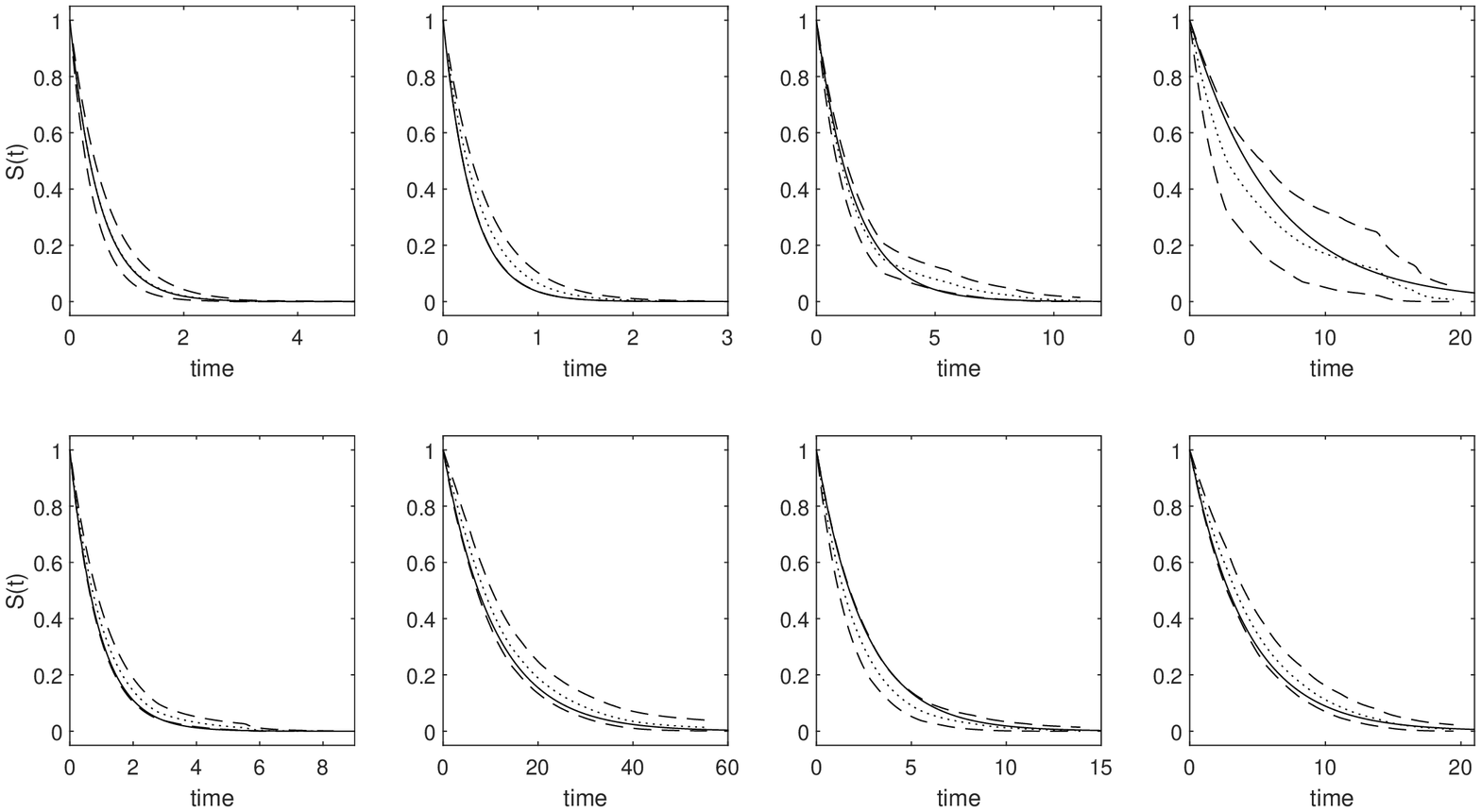}
\caption{The posterior mean (dotted) and 95\% credible intervals (dashed).  The survival curve corresponding to the median true hazard function of the data in each terminal node is also plotted (solid).}
\label{fig:sim7survival}
\end{center}
\end{figure}

\subsection{Prognostic \& predictive biomarker, comparison with BART, survival random forests, and Cox}
Identifying predictive and prognostic biomarker effects is important in drug development.  Sometimes there is a hypothesis that a certain sub-population, based on a predictive biomarker, will respond better to a treatment.   Other times the value of a prognostic biomarker is related to a patient's prognosis.  It is possible a biomarker may be both prognostic and predictive.

A single biomarker, $b$, and a treatment assignment indicator, $a$, are simulated for 1,000 subjects.  The biomarker is both prognostic and predictive, and there are two treatment arms. The time to event, $y$, is simulated from a Weibull model with random right-censoring in 7.8\% of subjects,
\begin{equation}
\begin{aligned}
    a &\sim Binomial(1, .5), \\
    b &\sim Uniform(0, 1), \\
    y &\sim Weibull(1 + 2ab, 1 + 5b),
\end{aligned}
\end{equation}
where $Weibull(\lambda, k)$ is a Weibull distribution with scale $\lambda$ and shape $k$.  For small biomarker values $b$ nearer to 0, there is a small treatment effect, whereas for values of $b$ closer to 1, there is a large treatment effect creating a predictive biomarker.  For a fixed treatment assignment $a$, the shape of the survival curve is a function of $b$, creating a prognostic effect for the biomarker.

The tree partition model was fit to the simulated data, and the tree with the highest posterior probability is plotted in Figure \ref{fig:pred_prog}. The tree and the estimated survival curves at each terminal node capture the predictive and prognostic effects.  The split rule of the root node indicates that there is a difference in survival functions between the two treatment arms.  Descending either right or left from the root node, all remaining splits are on the biomarker, indicating a prognostic biomarker effect.  For two patients with similar biomarker values, but different treatment assignments, their survival curves can be quite different.  For instance, a patient with a biomarker value of .8 has a much better survival curve with treatment assignment 1 verses 0; this indicates that the biomarker is also predictive.

The simulated data were fit to three additional models: a frequentist Cox proportional-hazard model with a non-parametric baseline hazard an an interaction between $a$ and $b$; Bayesian additive regression trees (BART) \citep{sparapani2016nonparametric}; and survival random forests \citep{ishwaran2008random}.  The survival curves, the 95\% confidence/credible intervals, and the true treatment effect for a patient with $a = 1$ and $b = .10$ are plotted in Figure \ref{fig:pred_prog_comparison}; for the survival random forest method, the 95\% intervals are calculated from the predictions of all trees.  The changing shape of the survival curve is not captured well in the Cox model since the proportional hazard assumption is too restrictive.  This is problematic; even if the Cox model identifies a significant treatment by biomarker interaction, which is often used to determine a predictive biomarker effect, the estimates of the survival curve may be inaccurate.  On the other hand, BART and random forests estimate the survival curves well, but often custom post-processing of the results in a potentially large-volume covariate space are required to better understand which predictors are prognostic, predictive, or both.  The THM automatically provides both clear interpretation and flexible estimation of the survival curves. 
\begin{figure}[t]
    \centering
    \includegraphics[scale = .55]{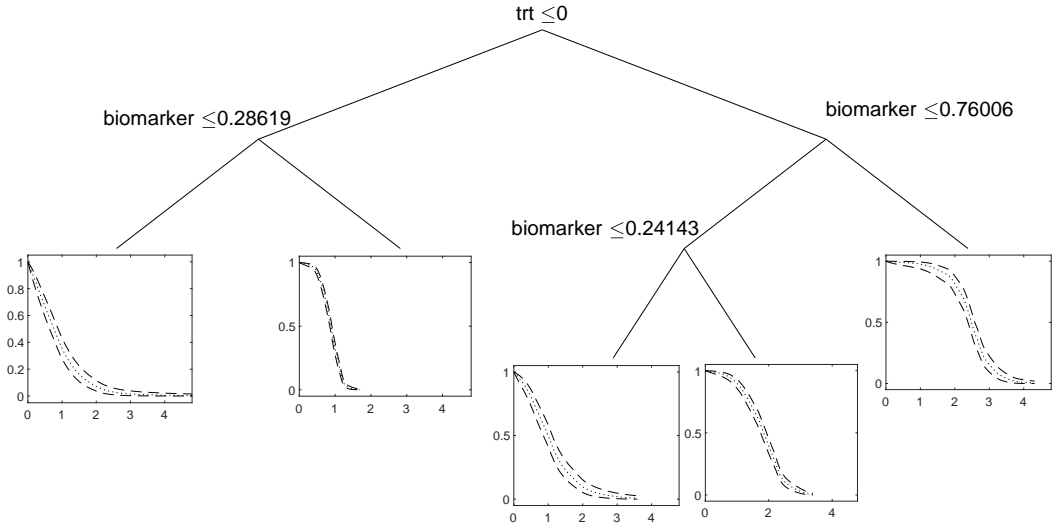}
    \caption{Tree with the highest posterior probability with the posterior mean (dotted) and 95\% credible intervals (dashed) at each terminal node.}
    \label{fig:pred_prog}
\end{figure}
\begin{figure}[t]
    \centering
    \includegraphics[scale = .55]{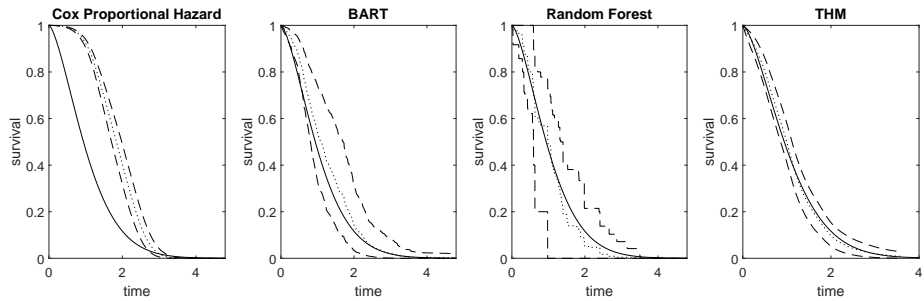}
    \caption{The true survival curve (solid) for a patient with $a = 1$ and $b = .10$, estimated mean (dotted) and 95\% confidence/credible intervals (dashed) for the frequentist Cox, BART, and THM.}
    \label{fig:pred_prog_comparison}
\end{figure}
\subsection{Analysis of Primary Biliary Cirrhosis Data}

Primary biliary cirrhosis survival data collected from 1974-1984 by the Mayo Clinic and obtained via the R survival package \citep{survival-package} were fit using the THM.  Covariates included treatment group (D-penicillamine or placebo), age, sex, edema status (0 no edema, 0.5 untreated or successfully treated, 1 edema despite diuretic therapy), serum bilirunbin (mg/dl), serum albumin (g/dl), and standardised blood clotting time.  Only subjects who were randomized to treatment and had no missing values were included in the analysis for a total of 276 subjects.  The outcome of interest was time to death, and patients were censored at the end of study or if they received a transplant, yielding a censoring rate of 60\%.

Ten-thousand MCMC samples were drawn and the tree with the highest posterior probability is shown in Figure \ref{fig:pbc} with the survival curves plotted at each terminal node.  The black curves represent the posterior mean and 95\% credible intervals from the tree model, and the red lines provide mean and 95\% confidence intervals from the Kaplan-Meier curves.  The plot indicates that the THM is very flexible, and close to the Kaplan-Meier estimates.

The tree indicates that patients with bilirubin levels less than 1.9 have the highest survival rates.  For subjects with no edema and bilirubin levels greater than 1.9 and less than or equal to 5.7, patients under the age of 45 had better survival.  Patients with bilirubin levels greater than 1.9 and edema had the worst survival.  These conclusions are consistent with the findings of \citet{dickson1989prognosis}, who used a Cox proportional hazards model on primary biliary cirrhosis data, however, the tree method provides greater flexibility in estimating the survival curves, especially in this case where the hazards are not proportional.  Last, we note that the tree did not split on treatment status, which is consistent with the findings of \citet{gong2004d} who performed a systematic review of the effect of D-penicillamine on patients with primary biliary cirrhosis and concluded that it does not decrease the risk of mortality.

\begin{figure}
    \centering
    \includegraphics[scale=.55]{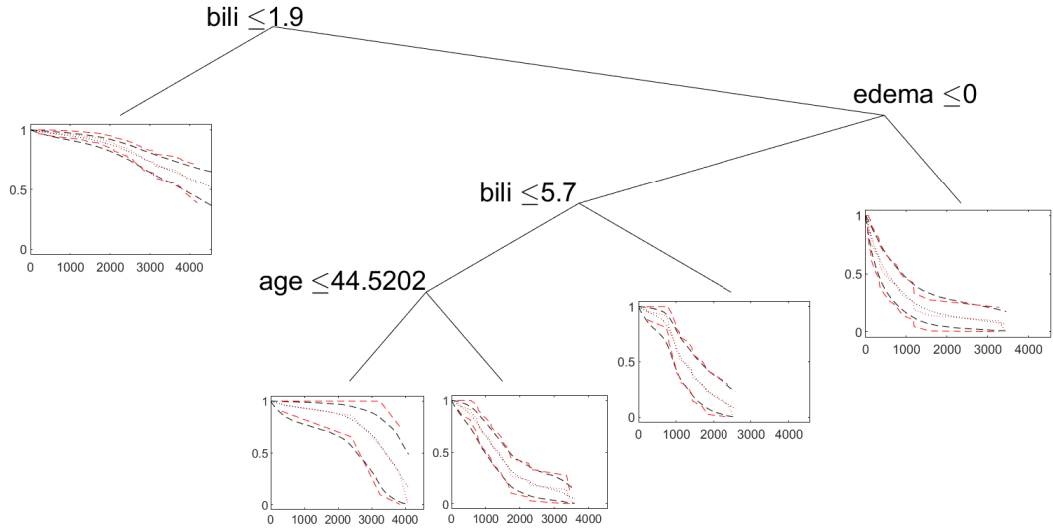}
    \caption{The tree with the highest posterior probability and the posterior mean (black, dotted), 95\% credible intervals (black, dashed), and the Kaplan-Meier estimates (red) at each terminal node.}
    \label{fig:pbc}
\end{figure}

\section{Conclusion}
The main advantage the tree hazards model (THM) is its ease in interpretation coupled with its flexibility.  Inference is easily obtained from inspecting posterior trees and the flexibly estimated survival curves at the terminal nodes.  In both simulations and applications, the model performs well in identifying interpretable partitions and flexibly estimating survival curves.  Furthermore, we have established weak consistency for the partition model under the proposed kernel, where we can accommodate flexible forms of the hazard function under moderate conditions.  


\bibliographystyle{apalike}
\bibliography{myreferences}

\newpage

\section{Appendix i: MCMC}
\subsection{Parallel tempering}
From their inception, Bayesian tree models have been notorious for their difficulty in mixing when using MCMC procedures.  Various methods have been proposed to overcome this issue including restarting the MCMC chain multiple times \citep{chipman1998bayesian}, parallel tempering \citep{gramacy2010categorical}, and multiset sampling \citep{leman2009}.

We implement parallel tempering \citep{geyer1991markov} modeled after \citet{gramacy2010categorical} to search the multi-modal posterior distribution and effectively search the parameter space.  The parallel tempering algorithm starts multiple MCMC chains in parallel, each with a slightly different target distribution, specifically
\[
p^{(j)}(T \mid \by) \propto \pi(T) \left(\prod_{i=1}^M p(\by_i)\right)^{t_j}
\]
where $1 = t_1 > t_2 > \ldots > t_d \geq 0$ are called the inverse temperatures, and are chosen using the sigmoidal ladder \citep{gramacy2010categorical}.  The algorithm proceeds by starting $d$ parallel chains with target distributions $p^{(1)}(T \mid \by),\ldots,p^{(d)}(T \mid \by)$.  Once every $l^\prime$ iterations, two adjacent chains are randomly selected and switch trees with a Metropolis-Hastings update with acceptance probability
\[
\alpha_2 = \min\left\{1, \frac{\tilde{p}^{(j)}(T^{j-1} \mid \by) \tilde{p}^{(j-1)}(T^{j} \mid \by)}{\tilde{p}^{(j)}(T^{j} \mid \by) \tilde{p}^{(j-1)}(T^{j-1} \mid \by)}      \right\}
\]
where $T^{j}$ is the current tree on the $j$th process and $T^{j-1}$ is the current tree on process $j-1$.  As before, $\tilde{p}^{(j)}(\cdot \mid \cdot)$ represents the Laplace approximation of the posterior.  The resulting chain for $\tilde{p}^{(1)}(T \mid \by)$ provides draws from the approximate posterior distribution and explores the parameter space more efficiently than a single reversible jump MCMC chain.


\subsection{Computational considerations} \label{sec:comp}
Although the Laplace approximation reduces the dimension of the parameter space which needs to be explored via MCMC, there are still a number of computational challenges.  The main computational cost in the algorithm comes from optimizing the covariance function parameters, $l_i$ and $\tau_i$, in each terminal node.  Optimizing these parameters requires finding  $\hat{\Bf}_i = \arg\max_{\Bf_i} p(\by_i \mid \Bf_i) \pi(\Bf_i)$ for each evaluated pair which involves solving multiple linear systems, e.g. in Newton's method, and a determinant.

Fortunately, the choice of the covariance function for $\omega_i(\cdot)$ along with a regular grid imply that we have a sparse closed-form solution for 
\[
\Sigma_i^{-1} = \frac{1}{\tau_i^2 (1 - \rho_i^2)} \begin{pmatrix}
1 & -\rho_i & & &   \\
-\rho_i & 1 + \rho_i^2 & -\rho_i &  &  \\
&  \ddots & \ddots & \ddots & \\
&  & -\rho_i & 1 + \rho_i^2 & -\rho_i \\
 & & & -\rho_i & 1
\end{pmatrix}
\]
where $\rho_i = \exp(-\Delta z / l_i)$ where $\Delta z$ is the distance between points of the regular grid.  This closed form implies that we never need to form the dense matrix $\Sigma_i$, and the matrices in \eqref{eq:laplace} can be stored in a sparse matrix form, reducing memory and computation demands. 


Lastly, although we use the same bins for each hazard function, we use a subset of the bins in each terminal node since the ranges of observed survival times will differ across partition elements.  This reduces the computation burden as well as avoids extrapolating beyond the terminal node's data.

\subsection{MCMC Details}
We now proceed by describing the acceptance ratio of the MCMC algorithm in more detail.  The acceptance proposal for a proposed tree in equation \eqref{eq:acceptance} is equivalent to

\[
\alpha = \min\left\{1,\frac{q(T \mid T^1) \tilde{p}(\by \mid T^1) \pi(T^1)}{q(T^1 \mid T) \tilde{p}(\by \mid T) \pi(T)}  \right\}
\]

where $q(\cdot \mid \cdot)$ is the proposal distribution from which $T^1$ was generated and $\tilde{p}(\by \mid T)$ is the approximate marginal likelihood.  The prior evaluation, $P(T)$, is given in \eqref{eq:prior}.  We now provide additional details regarding the evaluation of $q(\cdot\mid \cdot)$, which differ depending on which type of tree modification is proposed (i.e. grow, prune, change, swap).

\subsubsection{Grow}
The proposal probability for a grow step, $q_g(T^1\mid T)$, is $p_g (n_g m_i m^\star_{ij})^{-1}$ where $p_g$ is the probability of proposing a grow step, $n_g$ is the number of nodes on the current tree $T$ that are capable of a grow step, $m_i$ is the number of variables that have valid split rules on the $i$th node which was selected for a grow step, and $m^\star_{ij}$ is the number of split rules available on the $j$th variable on the $i$th node selected as the split variable.

\subsubsection{Prune}
The proposal probability for a prune step, $q_p(T^1\mid T)$, is $p_p (n_d)^{-1}$ where $p_p$ is the probability of proposing a prune step and $n_d$ are the number of interior nodes for which a prune step is possible (i.e. parent nodes which have both children as terminal nodes).  

\subsubsection{Change}
The proposal probability of a change step, $q_c(T^1\mid T)$ is $p_{ch} n_c^{-1} p_{\text{prior}} (m_i m^\star_{ij})^{-1}$ where $p_{ch}$ is the probability of proposing a change step, $n_c$ is the number of interior nodes eligible for a change step (which is all interior nodes if one includes the trivial change step of keeping the same rule if there is only one rule available to a particular node), and $p_{\text{prior}}$ is the probability of drawing a rule from the prior.  Note that if the current rule is based on a categorical variable,  $p_{\text{prior}} = 1$ and $m_i$, $m^\star_{ij}$ are the same as those found in $q_g(\cdot,\cdot)$.  If however, the current rule is a continuous variable, $p_{\text{prior}} \in [0,1]$ is the value pre-specified by the user.  If the new rule is proposed by the prior, $m_i$, $m^\star_{ij}$ are the same as those found in $q_g(\cdot,\cdot)$.  If, however, the new rule is chosen by selecting the next largest or smallest rule at random from the current continuous variable on which the current rule is based, then $m_i = 1$ and $m^\star_{ij} = 2$ unless the current rule is on the boundary of available rules or is the only rule available, then $m^\star_{ij} = 1$.

\subsubsection{Swap}
The proposal probability of a swap step, $q_s(T^1\mid T)$ is $p_s n_s^{-1}$ where $p_s$ is the probability of proposing a swap step and $n_s$ is the number of possible swap steps available on the current tree, $T$.

\bigskip

Prior to running the reversible jump MCMC algorithm, the user specifies $p_g,p_p,p_{ch}$ and $p_s$.  However, on a given MCMC iteration, these values may need to be adjusted.  For instance, if one starts with the root node, the only possible proposal is a grow step, therefore $p_g = 1$ and $p_p = p_{ch} = p_s = 0$. Other examples (not necessarily an exhaustive list) include the fact that a swap step is not possible with fewer than 3 terminal nodes, and birth steps may not be possible if the tree grows large enough that splitting any terminal node would result in a terminal node with fewer than the minimum number of required observations required by the user.  Therefore, on each iteration of the MCMC chain, one must check to determine which proposals are actually available and adjust the probabilities accordingly.  This is also the case when calculating the probability of returning to the current tree from the proposed tree, $q(T \mid T^1)$.

\section{Appendix ii: Proof of the results} \label{sec:proofs}
\subsection{Proof of Proposition \ref{max_sup}}
Note that for $\beta>\beta'>0$, by $[A5]$, $P(T_n\geq n^{\beta'}) \leq ne^{-c n^{\alpha\beta'}}$ for a constant $c>0$, and \[P(T_n\geq n^{\beta'} \text{ infinitely often })\leq lim_{k\uparrow \infty}\sum_{n=k}^\infty ne^{-cn^{\alpha\beta'}}= 0.\]
Hence, the result follows.
\subsection{Proof of Lemma \ref{gp_support}}
We have  $\tilde{{w}}_{0,i}^n(t)=\int_0^{T_n}\tilde{g}_i(s)K(s/l_i^*,t/l_i^*)ds$ for $0\leq t\leq T_n$ and let $\hat{{w}}_{0,i}^n(t)=\frac{c_\xi}{T_n^2}\sum_{j=1}^{\lceil{T^3_n/c_\xi}\rceil}\tilde{g}(s_j)\tau^2_iK(s_j/l_i,t/l_i)$ be its approximation by a Riemann-sum type expansion using $\frac{c_\xi}{T_n^2}$ length grids $s_1=0,s_j-s_{j-1}=\frac{c_\xi}{T_n^2}$, and so on  and hyper-parameters are $l_i,\tau^2_i$.  Without loss of generality, we assume $T_n>1$. 

We have: $|\tilde{{w}}_{0,i}^n(t)- \hat{{w}}_{0,i}^n(t)|\leq \sum_{j=1}^{\lceil{T^3_n/c_\xi}\rceil}\large(\int_{s_j}^{s_{j+1}}|\tilde{g}_i(s)| |K(s/l_i^*,t/l_i^*)-K(s/l_i,t/l_i)|ds+\int_{s_j}^{s_{j+1}}|\tilde{g}_i(s)-\tilde{g}_i(s_j)|K(s/l_i,t/l_i)ds+\int_{s_j}^{s_{j+1}}|\tilde{g}_i(s_j)||K(s/l_i,t/l_i)-K(s_j/l_i,t/l_i)|ds+|\tau^2_i-1|\frac{c_\xi}{T_n^2}\sum_{j=1}^{\lceil{T^3_n/c_\xi}\rceil}|\tilde{g}(s_j)|K(s_j/l_i,t/l_i)\large)$.  Here $l_i\in[l_i^*,l_i^*+c_\xi/T^3_n)$ and $\tau^2_i\in[1,1+c_\xi/T^2_n)$ and  $c_\xi>0$ is a small constant such that 
$|\hat{w}_{0,i}^n(t)-w_{0,i}(t)|<\frac{\xi}{2T_n}$, $0\leq t \leq T_n$.

Note that  $\|\hat{w}_{0,i}^n\|^2_{H_{l_i,n}}<T_n^2M_1$, where ${H_{l_i,n}}$ is the {\it Reproducing Kernel Hilbert Space} (RKHS) of a Gaussian process on the interval $[0,T_n]$ with covariance kernel $K_{l_i}(s,t)=K(s/l_i,t/l_i)$ and $\|h\|^2_{H_{l_i,n}}$ is the squared  RKHS norm of $h\in H_{l_i,n}$ and $M_1$ is a constant. The  RKHS is the closer of linear combinations $\sum_{j=1}^{k} a_jK_l(t_j,t)$, $t_j,t \in [0,T_n]$,  under the Hilbert space norm induced by  inner products $\langle\sum_{j=1}^{k} a_jK_l(t_j,\cdot),\sum_{j'=1}^{k'} b_{j'}K_l(s_{j'},\cdot)\rangle=\sum_{j,j'}a_jb_{j'}K_l(t_j,s_{j'})$; $s_{j'},t_j\in[0,T_n]$. 

From [A6], $-\log\pi(l_i\in [l_i^*,l_i^*+c_\xi/T^3_n))=o(n)$, $-\log\pi(\tau^2\in [1,1+c_\xi/T^2_n))=o(n)$.

Now \[-\log\pi(S^i_{\xi,l_i,\tau^2_i})\leq \|\hat{w}_{0,i}^n\|^2_{H_{l_i,n}}-\log(P(\|w^n\|_\infty<\frac{\xi}{2T_n}),\]
where $w^n$ follows Gaussian process on $[0,T_n]$ with covariance kernel $\tau_i^2K(s/l_i,t/l_i)$, and $S^i_{\xi,l_i,\tau^2_i}$ denotes the set $S^i_\xi$ for fixed hyper-parameter values $l_i, \tau_i^2$. The above inequality can be  found  in the Gaussian process literature and Gaussian process density estimation applications (for example in \citet{van2008reproducing}).

Next, we use a conservative  lower bound for the small ball probability for a Gaussian process using the results from \citet{aurzada2008small}.  From, the fact $\rho^2(s,t)=E[(w_n(s)-w_n(t))^2]\sim |s-t|$ for $|s-t|$ sufficiently small. The covering number of $[0,T_n]$ by the pseudo-metric $\rho$, by $\xi/T_n$ radius ball is of the order of $T_n^3=o(n)$ (Proposition \ref{max_sup}) and satisfies the conditions of Theorem 1 from  \citet{aurzada2008small} and hence, $-\log(P(\|w^n\|_\infty<\xi/(2T_n))=o(n)$.

Together combining the parts, $-\log\pi(S^i_\xi)=o(n)$.

%

\subsection{Proof of Lemma \ref{kl_sup}}
Note that  $h_{0,i}$ are bounded away from zero and  infinity and suppose, $0< m_0< h_{0,i}\leq M_0<\infty$. 
Expanding $KL(p^*,p)$, for  ${\bf x}\in  {\hat{\mathscr{P}^*}\cap {\mathscr{P}^*}}$, it boils down to showing for universal  constants $c_1,c_2,c_3>0$,
\begin{eqnarray}
  |\int_0^\infty p_{0,i}(t)\bar{F}_c(t)\log\frac{h_{0,i}(t)}{h_i(t)}dt|<c_1\epsilon; \nonumber \\
 \int_0^\infty p_{0,i}(t) \bar{F}_c(t)|H_i(t)-H_{0,i}(t)|dt<c_2\epsilon; \nonumber \\ \int_0^\infty p_c(t) \bar{F}_{0,i}(t)|H_i(t)-H_{0,i}(t)|dt<c_3\epsilon,
 \label{to_show}
\end{eqnarray}
where $ \bar{F}_{0,i}(t)=1-F_{0,i}(t),\   \bar{F}_{c}(t)=1-F_{c}(t)$. 

By construction, $\log h_i(t)$ is bounded over $[0,\infty)$. We can choose $T_0>0$ in the support,  such that  $\int_{T_0}^\infty [p_{0,i}(t)+p_c(t)]H_{0,i}(t)dt<M_0\int_{T_0}^\infty t[p_{0,i}(t)+p_c(t)]dt<\epsilon$, and   $\int_{T_0}^\infty [p_{0,i}(t)+p_c(t)]H_{i}(t)dt<\epsilon$, and $|\int_{T_0}^\infty p_{0,i}(t)\bar{F}_c(t)\log\frac{h_{0,i}(t)}{h_i(t)}dt|<\epsilon$. 

For large $n$, $T_n\geq T_0$ with probability one. We can choose small $\xi>0$ such that,  for $0\leq t\leq T_n$,  $|h_{0,i}-h_i(t)|<|e^{2\xi}-1|<\xi'$, and  $\int_0^{T_n} p_{0,i}(t)log\frac{h_{0,i}}{h_i}dt<\epsilon$, we have for $0\leq t\leq T_n$,  $|H_{0,i}(t)-H_i(t)|<2k_0\xi' t$; $k_0>0$ a generic universal constant. Also, $\xi'$ is small enough such that  $E_{p_{0,i}}[2k_0\xi' t],E_{p_c}[2k_0\xi' t]<\epsilon$. Hence, $ \int_0^{T_n} p_{0,i}(t) \bar{F}_c(t)|H_i(t)-H_{0,i}(t)|dt<\epsilon,  \int_0^{T_n} p_c(t) \bar{F}_{0,i}(t)|H_i(t)-H_{0,i}(t)|dt<\epsilon$.

Hence, the conditions in  equation \eqref{to_show} are satisfied for ${\bf x} \in  {\hat{\mathscr{P}^*}\cap {\mathscr{P}^*}}$. Next, we consider ${\hat{\mathscr{P}^*}\Delta {\mathscr{P}^*}}$.

  Note that,  $h_{i}$ are bounded away from zero and  infinity, for any $\xi>0$, and $H_i(t),H_{0,i}(t)\preceq t$.  Note that for observed $y_j$, $y_j=y_j^*$ if the corresponding observation is not censored and $y_j\leq y_j^*$, where $y_j^*\sim p_{0,i'}(t)$ is the uncensored observation (potentially not observed), for some $i'\in\{1,\cdots,m^*\}$. $E_{p_{0,i'}}[H_{0,i}(y)]\leq E_{p_{0,i'}}[H_{0,i}(y^*)]<C_0'<\infty$, as  $E_{p_{0,i'}}[t]<\infty$. Also,   $E_{p_c}[t]<\infty$.
  From the fact $d(\mathscr{P}^*,\hat{\mathscr{P}}^*)<\epsilon$,  for $j\neq j'$,  $\int_{\hat{\mathscr{P}^*}\Delta {\mathscr{P}^*}} \int p_{0,j}(t) \bar{F}_{c}(t) |\log\frac{h_{0,j}(t)}{h_{j'}(t)}|$  ${\bf I}_{x\in P_j^*\cap{\hat{P}_{j'}^*}}dt d\mu_x({\bf x})<k_1\epsilon$, $k_1>0$, and  $\int_{\hat{\mathscr{P}^*}\Delta {\mathscr{P}^*}} \int p_{0,j}(t) \bar{F}_{c}(t) |H_{j'}(t)-H_{0,j}(t)|$  ${\bf I}_{x\in P_j^*\cap{\hat{P}_{j'}^*}}dt d\mu_x({\bf x})<k_2\epsilon$,  and $\int_{\hat{\mathscr{P}^*}\Delta {\mathscr{P}^*}} \int $  $p_{c}(t)  \bar{F}_{0,j}(t) |H_{j'}(t)-H_{0,j}(t)|$  ${\bf I}_{x\in P_j^*\cap{\hat{P}_{j'}^*}}dt d\mu_x({\bf x})<k_2\epsilon$, where $k_1,k_2>0$ are constants; then \newline $\int_{\hat{\mathscr{P}^*}\Delta {\mathscr{P}^*}}\int  p^* \log\frac{p^*}{\hat{p}}d(t,\delta)d\mu_x({\bf x})<k_3\epsilon$,  where $k_3>0$ is a constant.
  
  Hence, combining the parts from ${\hat{\mathscr{P}^*}\Delta {\mathscr{P}^*}}$ and $ {\hat{\mathscr{P}^*}\cap {\mathscr{P}^*}}$ the result follows.

\subsection{Proof of Proposition \ref{tree_support1}}
Suppose, $(v_i^*,r_i^*);\ i=1,\cdots$ be associated with true splits, where $v^*_i$ denotes the covariate corresponding $i$th split and $r^*_i$ be the split location. Let $(v_i^*,\hat{r}_i^*);\ i=1,\cdots$ be associated with  approximating splits, where $ \hat{r}_i^*=r_i^*$ for $v^*_i$ discrete and $\hat{r}_i\in(r^*_i\pm c_\delta)$ for some constant $c_\delta>0$ for  $v^*_i$ continuous.  For $v^*_i$ continuous, choosing $c_\delta$ sufficiently small, we have $\mu_x(\mathscr{P}\Delta \hat{\mathscr{P}}^*)<\xi$. From the fact, we can select $\hat{r}_i$ in $(r^*_i\pm c_\delta)>0$ for large $n$, and  $-\log P(\hat{r}_i\in(r^*_i\pm c_\delta))=o(n)$, the result follows
(condition [A4]).

\subsection{Proof of Theorem 1}
It is enough to prove Theorem \ref{thm1} for a neighborhood of the form  $\{p:\int g p^*(\cdot)-\int g p(\cdot)<\epsilon\}$ instead. We redefine $U^{\epsilon}_g=\{p:\int g p^*(\cdot)-\int g p(\cdot)<\epsilon\}$ and note that $\{p:|\int g p^*(\cdot)-\int g p(\cdot)|<\epsilon\}$ is the intersection of $U^{\epsilon}_g$ and $U^{\epsilon}_{-g}$. We have tests $\phi_n$, $0\leq \phi_n\leq 1$, such that $\max\{E_{p^*}[\phi_n], sup_{p\in {U^{\epsilon}_g}^c }E_p[1-\phi_n] \}<e^{-cn\epsilon}$, $c>0,\epsilon>0$. See Remark 4.4.1 in \citet{ghosh2003bayesian} for details.

For the denominator  in the following equation \eqref{postcal1}, we cannot use Fatou's lemma to establish the lower bound, as the domain set of a Gaussian Process  prior changes over $n$. We instead proceed as follows.
First, we consider ${p}(\cdot)$'s within small KL distance  from the  $p^*(\cdot)$.  We then show that the  negative log of prior probability of that small KL ball is of the order of $o(n)$ and using direct calculation in the denominator, we establish a bound on the the difference of the average log-likelihood ratio.
In particular,
\begin{equation}
\Pi_n({U^{\epsilon}_g}^c|.)\leq\frac{\int_{{U^{\epsilon}_g}^c}\prod_{j=1}^n\frac{p(y_j,\delta_j,{\bf x}_j)}{p^*(y_j,\delta_j,{\bf x}_j)}d\pi(\cdot)}{\int_{U^{\epsilon_n}}\prod_{j=1}^n\frac{p(y_j,\delta_j,{\bf x}_j)}{p^*(y_j,\delta_j,{\bf x}_j)}d\pi(\cdot)}=\frac{N_n}{D_n}.
\label{postcal1}
\end{equation}
Let $\mathscr{P}^*$ be the true partition and $\hat{\mathscr{P}^*}$ be such that $d(\mathscr{P}^*,\hat{\mathscr{P}^*})<\epsilon'$. Let   $\hat{\mathscr{P}^*}=\{\hat{P}_1^*,\dots,\hat{P}_{m^*}^*\}$ and  $\hat{P}^*_i$  approximating   $P^*_i$.
In equation \eqref{postcal1}, let ${U^{\epsilon_n}}=U^{\epsilon_n}_{\hat{\mathscr{P}}^*}$ be the set of ${p}$ with  corresponding  hazard functions $\hat{h}_i$'s for their corresponding ${ {w}}_{i}(t)$'s, $\beta_i$'s,  and ${ {w}}_{i}(t)$'s  are in $ \frac{\epsilon'}{T_n }$ supremum neighborhood around $\tilde{ {w}}_{0,i}^n(t)$, in $[0,T_n]$ and $\beta_{i}$ in $(\beta_{0,i}\pm \frac{{\epsilon'}}{T_n})$ for some $\epsilon'>0$ over $\hat{P}^*_i$.   Note that $\hat{h}_i(t)=\hat{h}_i(T_n)$ for $t>T_n$. Also,  $\frac{p(y_j,\delta_j,{\bf x}_j)}{p^*(y_j,\delta_j,{\bf x}_j)}=\frac{p(y_j,\delta_j|{\bf x}_j)}{p^*(y_j,\delta_j|{\bf x}_j)}$. Here $w^n_{0,i}(t)=\int_0^{T_n}\tilde{g}_i(s)K(l^*_is,l^*_it)ds$.


Let $\hat{p}^*$ be the density under partition $\hat{\mathscr{P}^*}$ using $w_{0,i}$, $\beta_{0,i}$'s for $\hat{P}_i^*$  and let $\hat{p}$ be the density for the partition $\hat{\mathscr{P}^*}$, using   $\beta_{0,i}$'s and  $\tilde{w}^n_{0,i}$'s in $\hat{P}^*_i$.

%
Note that  from equation \eqref{postcal1}
\begin{eqnarray}
\prod_{j=1}^n\frac{p(y_j,\delta_j,{\bf x}_j)}{p^*(y_j,\delta_j,{\bf x}_j)}=\large(\prod_{j=1}^n\frac{p(y_j,\delta_j,{\bf x}_j)}{\hat{p}(y_j,\delta_j,{\bf x}_j)}\large)\large(\prod_{j=1}^n\frac{\hat{p}(y_j,\delta_j,{\bf x}_j)}{\hat{p}^*(y_j,\delta_j,{\bf x}_j)}\large)\large(\prod_{j=1}^n\frac{\hat{p}^*(y_j,\delta_j,{\bf x}_j)}{{p}^*(y_j,\delta_j,{\bf x}_j)}\large)=D^{(1)}_n\times D^{(2)}_n\times D^{(3)}_n.
\label{decomp}
\end{eqnarray}

  Now, $-\log(\pi(\hat{\mathscr{P}^*}))=o(n)$ and $-\log(\pi(\beta_{i}\in (\beta_{0,i}\pm \frac{{\epsilon'}}{T_n})))=o(n)$.  Hence,  using Lemma \ref{gp_support}, we have $-\log(\pi (U^{\epsilon_n}))=o(n)$ and therefore, choosing small enough $\epsilon'<\epsilon$ for large $n$, $e^{\frac{cn\epsilon}{4}}{e^{\log(D_n)}}>1$ using Lemma \ref{dc_prop} (given after this proof).  Hence, for large $n$
 \begin{eqnarray*}
E_{p^*}[\Pi_n({U^{\epsilon}_g}^c|.)]&=&E_{p^*}[\phi_n\Pi_n({U^{\epsilon}_g}^c|.)]+E_{p^*}[(1-\phi_n)\Pi_n({U^{\epsilon}_g}^c|.)]\\
    &\leq&E_{p^*}[\phi_n] +e^{\frac{cn\epsilon}{4}}\int_{p\in{U^{\epsilon}_g}^c} E_p[(1-\phi_n)]d\pi\\
    &\leq&2e^{-cn\epsilon/2}
\end{eqnarray*}
Hence,  $ \sum_n P_{p^*}( \Pi_n({U^{\epsilon}_g}^c|.)>e^{-cn\epsilon/4}) <\infty$ and hence by the Borel-Cantelli Lemma $\Pi_n({U^{\epsilon}_g}^c|.)  \rightarrow 0$ almost surely.

\begin{lemma}
From equation \eqref{decomp}, $n^{-1}|\log D_n^{(1)}|,\ n^{-1}|\log D_n^{(2)}|,\ n^{-1}|\log D_n^{(3)}|<k\epsilon'$ for large $n$ almost surely, for some universal constant $k>0$, uniformly over $p$ in  ${U^{\epsilon_n}}$.
\label{dc_prop}
\end{lemma}
\begin{proof}

This result follows from  the construction of ${U^{\epsilon_n}}$. Let $\hat{h}_{0,i}$ be the hazard function corresponding to $\beta_{0,i}$ and $\hat{w}_{0,i}^n$ for $\hat{p}$. Similarly we can define $\hat{H}_{0,i}$. Let $\hat{h}_{i}$, $\hat{H}_i$ be the functions for the $i$th partition $\hat{P}^*_i$ for a generic  ${p}$ in $U^{\epsilon_n}$. Without loss of generality, we assume $T_n\geq 1$ for large $n$. By construction, $h_i,\hat{h}_{0,i}$'s are bounded away from zero and infinity. Assume, $0< m'_0< \hat{h}_{0,i}\leq M'_0<\infty$.

\vspace{0.1in}

\noindent{\underline{\it Calculation for $|n^{-1}\log D_n^{(1)}|,|n^{-1}\log D_n^{(2)}|$}}

By the construction of  ${U^{\epsilon_n}}$, for $p(\cdot) \in {U^{\epsilon_n}}$, $|\log(\frac{{p}(y_j, \delta_j,{\bf x}_j)}{\hat{{p}}(y_j, \delta_j,{\bf x}_j)})|\leq C\epsilon'$ for some universal constant $C>0$. This step follows as, $|\log \frac{h_{0,i}}{h_i}|<\frac{2\epsilon'}{T_n}$; $|\hat{h}_{i}(t)-\hat{h}_{0,i}(t)|< M_0'|e^{ \frac{2\epsilon'}{T_n}} -1|\leq c_1 \frac{\epsilon'}{T_n}$; and $|\hat{H}_{i}(t)-\hat{H}_{0,i}(t)|\leq c_1\epsilon'$, $c_1>0$ a constant, $0\leq t\leq T_n$. 
Hence, $|n^{-1}\log D_n^{(1)}|\leq n^{-1}\sum_{j=1}^n\sum_{i=1}^{m^*}\large[|\log \frac{\hat{h}_{0,i}(y_j)}{\hat{h}_i(y_j)}|+|\hat{H}_{0,i}(y_j)-\hat{H}_i(y_j)|\large]{\bf I}_{{\bf x}_j\in \hat{P}^*_i} <C \epsilon'$, $C>0$ some constant. Note that the bound is uniformly over $p$ in $U^{\epsilon_n}$.

 Similarly  for $n^{-1}\log D_n^{(2)}$, by [A3], $|\hat{h}_{0,i}(t)-{h}_{0,i}(t)|\preceq \frac{\epsilon'}{T_n}$ and $|\hat{H}_{0,i}(t)-{H}_{0,i}(t)|\preceq\epsilon'$ for $0\leq t\leq T_n$. Here $\preceq$ implies `less than equal to' of  a constant multiple.  Hence,  $|\log(\frac{\hat{p}(y_j, \delta_j,{\bf x}_j)}{\hat{{p}}^*(y_j, \delta_j,{\bf x}_j)})|\leq C\epsilon'$. Here, without loss of generality we can use the same constant as in    the ratio in $n^{-1}\log D_n^{(1)}$. Hence,  $|n^{-1}logD_n^{(1)}|$, $|n^{-1}logD_n^{(2)}|<C \epsilon'$.

 \vspace{0.1in}

\noindent {\underline{\it Calculation for $n^{-1}|\log D_n^{(3)}|$}}


The absolute value of the term involving $h_{0,i}(\cdot)$'s in $n^{-1}\log D_n^{(3)}$ can be bounded in the following manner. 
Let \[S_{h}= n^{-1}|\sum_{j,i} [\log h_{0,i}(y_j){\bf I}_{\delta_j=1}{\bf I}_{{\bf x}_j\in {P}^*_i}-\sum_{j,i}\log h_{0,i}(y_j){\bf I}_{\delta_j=1}{\bf I}_{{\bf x}_j\in \hat{P}^*_i}]|.\] Then
 $S_h\leq C^* n^{-1}\sum_{j=1}^n {\bf 1}_{{\bf x}_j\in \mathscr{P}^*\Delta \hat{\mathscr{P}^*}}<2C^*\epsilon'$ for large $n$ with probability one, where $C^*=2\max_{i,t}|\log h_{0,i}(t)|$, and $E[S_h]\leq C^*\epsilon'$ .
  
 Next, we consider the terms involving $H_{0,i}(\cdot)$'s, that is  $n^{-1}\sum_{j=1}^n\sum_i\Large[ H_{0,i}(y_j){\bf I}_{{\bf x}_j\in P^*_i}-H_{0,i}(y_j){\bf I}_{{\bf x}_j\in \hat{P}^*_i}\Large]$. Note that  $E_{p_{0,i'}}[H_{0,i}(y)]<C_0'<\infty$ for some constant $C_0'$ for all $i,i'$, as shown in Lemma \ref{kl_sup}. 

  Let $n_i$ be the number of observations for which ${\bf x} \in P^*_i$, similarly we define $\hat{n}_i$'s for $\hat{P}^*_i$'s. Then $n_i/n\rightarrow \mu_x(P_i^*)$ and  $\hat{n}_i/n\rightarrow \mu_x(\hat{P}_i^*)$ almost surely.   For, $H^*(y,x)=H_{0,i}(y)$ if ${\bf x}\in P^*_i$ and $\hat{H}^*(y,x)=H_{0,i}(y)$ if ${\bf x}\in \hat{P}^*_i$, $E_{p^*}[\hat{H}^*(y,x)]<\infty$ and $ E_{p^*}[{H}^*(y,x)]<\infty$. 
    For any $\epsilon'>0$, ${n}^{-1}|\sum_{j=1}^n\large(\sum_iH_{0,i}(y_j){\bf I}_{{\bf x}_j\in P^*_i}-E_{p^*}[{H}^*(y,x)]\large)|<\epsilon'/4$ for large $n$ almost surely and ${n}^{-1}|\sum_{j=1}^n$ $\large(\sum_iH_{0,i}(y_j){\bf I}_{{\bf x}_j\in \hat{P}^*_i}-E_{p^*}[\hat{H}^*(y,x)]\large)|<\epsilon'/4$ for large $n$ almost surely. 
    
   Again, $H^*(y,x)$, $\hat{H}^*(y,x)$ are equal for ${\bf x}\in \hat{P}^*_i\cap P^*_i$ for $i=1,\dots,m^*$ and  $|E_{p^*}[{H}^*(y,x)-E_{p^*}[\hat{H}^*(y,x)]|\leq 2C_0'd(\mathscr{P}^*,\hat{\mathscr{P}}^*)<2C_0'\epsilon'$; $C_0'>0$. This proves our claim.
  


\end{proof}

\subsection*{Proof of Theorem \ref{thm2}}
Let $H_{1,i}^n$ be the subset of RKHS with norm less than or equal to 1 for the  Gaussian process $w_i(\cdot)$ in $[0,T_n]$ for a fixed hyper parameter value $l_i$ in the support. Let $\epsilon_2 B^n_{1,i}$ be the Banach ball of radius $\epsilon_2$. Then, for any $\epsilon_2>0$,
\[-\log P(w_i\not \in M_nH_{1,i}^n+\epsilon_2  B_{1,i}^n ) \succeq M_n^2\] from Borel's inequality.

Let $K_{n,i}=\{ w_i:w_i\in M_nH_{1,i}^n+\epsilon_2 B_{1,i}^n\}$ and $K_n=\{w_i\in K_{n,i}\forall i \text{ and } \#(T)\leq D_n\}$. For $\epsilon_1>0,\alpha_1>0$,  for each partition $i$, we consider  the covering  of $K_{n,i}$ with sup-norm ball of radius $\alpha_1\epsilon^2_1/T_n$, and the covering of $\beta_i$'s by $\alpha_1\epsilon^2_1/T_n$ ball for $K_\beta=\{\beta_i:|\beta_i|\leq M_n^{2/\alpha'}; \#(T)\leq D_n\}$, where $\alpha'>0$ is from $B4$. For a particular tree $T$, we can define $K_n^T$ and $K_\beta^T$ similarly, and $K_{n,\beta}^T$ be their product space,  $K_{n,\beta}^T=\{(\beta_i,w_i):w_i\in M_nH_{1,i}^n+\epsilon_2 B_{1,i}^n, |\beta_i|\leq M_n^{2/\alpha'}\}$.  Let $K_{n,\beta}=\cup_{T:\# (T)\leq D_n}K_{n,\beta}^T$. Choosing $\alpha_1>0$ sufficiently small, for any $\epsilon_1>0$, this will induce a $\epsilon_1$ radius $L_1$ covering of the densities in $K_{n,\beta}$, from the relationship between Kullback-Leibler and total variation distance. Also, $-\log \pi(K_{n,\beta}^c)\succeq -D_n (\log D_n + \log m_n) +M_n^2 \succeq M_n^2$, as $-\log \pi(K_{n,i}^c)\succeq M_n^2.$ 

Now, the log-covering number of $K_{n,i}$ in sup norm, for fixed $l_i$, $ \log N(\alpha_1\epsilon^2_1/T_n,K_{n,i},\|\cdot\|_\infty)\preceq M_n^2+n^{\tilde{\delta}} $ for any $\tilde{\delta}>0$  from the the entropy calculation from Theorem 2.1 in  \cite{van2008rates} and the small ball probability calculation in Lemma \ref{gp_support}, and therefore, $\log N(\epsilon_1,K_{n,\beta},\|\cdot\|_1)\preceq D^{1+\delta'}_n (M_n^2+n^{\tilde{\delta}})\log m_n =o(n)$, by choosing $\tilde{\delta}>0$ sufficiently small. This log covering bound holds over the support of $\pi(l_i)$'s , by covering the compact support  of $\pi(l_i)$' as it is supported on  $M_n^{-d'}$ spaced grids and the log covering number of the support of hyper parameter over all possible partition  is $o(n)$.   Together, taking union over the choice of $l_i$'s over the terminal nodes,  we denote the sieves by $K_{n,\beta,l}$ where the log $L^1$-covering number of $K_{n,\beta,l}$  is $o(n)$, and the negative  log  prior probability outside $K_{n,\beta,l}$ is at least   $O(M_n^2)$.

As a result, we can have test statistics \citep{ghosh2003bayesian} $\phi_n$'s such that
\[ E_{p^*}[\phi_n]\leq e^{-cn}\text{ and } sup_{p\in U^c_\epsilon\cap K_{n,\beta,l}}E_p[1-\phi_n]\leq e^{-cn}; c>0.\]


Let  ${\tilde{U}^{\xi_n}}=\tilde{U}^{\xi_n}_{\hat{\mathscr{P}}^*}$ be the set of ${p}$ with  corresponding  hazard functions $\hat{h}_i$'s for their corresponding ${ {w}}_{i}(t)$'s, $\beta_i$'s,  and ${ {w}}_{i}(t)$'s  are in $ \frac{\xi_n}{T_n }$ supremum neighborhood around $\tilde{ {w}}_{0,i}^n(t)$, in $[0,T_n]$ and $\beta_{i}$ in $(\beta_{0,i}\pm \frac{{\xi_n}}{T_n})$ over $\hat{P}^*_i$ where  $\xi_n \sim M_n^2/n$ specified later, and $\hat{\mathscr{P}}$ corresponds to tree $\hat{T}^*$, with same number of terminal nodes as $T^*$, and $d(\mathscr{P}^*,\hat{\mathscr{P}}^*)\preceq n^{-\delta}$ for some $\delta>0$.    As $m_n$ goes to infinity in a polynomial order of $n$ (that is $m_n\sim n^\delta$, for some $\delta>0$), for true  tree $T^*$ and partition $\mathscr{P^*}$, we have an approximation with same number of nodes and partition $\hat{\mathscr{P}}^*_n$ such that $d(\hat{\mathscr{P}}^*_n,{\mathscr{P}}^*)\preceq n^{-\delta}$.  Note that   $\frac{p(y_j,\delta_j,{\bf x}_j)}{p^*(y_j,\delta_j,{\bf x}_j)}=\frac{p(y_j,\delta_j|{\bf x}_j)}{p^*(y_j,\delta_j|{\bf x}_j)}$, and  $w^n_{0,i}(t)=\int_0^{T_n}\tilde{g}_i(s)K(l^*_is,l^*_it)ds$.

Finally,
\begin{eqnarray}
E_{p^*}[\Pi_n(U_\epsilon^c|\cdot)]&=&E_{p^*}[\phi_n \Pi_n(U_\epsilon^c|\cdot)]+E_{p^*}[(1-\phi_n){\bf 1}_{K_{n,\beta,l}}\frac{\tilde{N}_n}{\tilde{D}_n}]+E_{p^*}[(1-\phi_n){\bf 1}_{K^c_{n,\beta,l}}\frac{\tilde{N}_n}{\tilde{D}_n}]\nonumber\\
&\leq &e^{-cn}+ E_{p^*}[(1-\phi_n){\bf 1}_{K_{n,\beta,l}}\frac{e^{ \epsilon'' n}\tilde{N}_n}{e^{ \epsilon'' n}\tilde{D}_n}]+E_{p^*}[(1-\phi_n){\bf 1}_{K^c_{n,\beta,l}}\frac{e^{ \epsilon'' M_n^2}\tilde{N}_n}{e^{ \epsilon'' M_n^2}\tilde{D}_n}]\nonumber\\
&\leq &e^{-cn}+e^{-cn}{e^{ \epsilon'' n}}+e^{-c^*M_n^2}{e^{ \epsilon'' M_n^2}}.
\label{thm2_eq}
\end{eqnarray}

Here $\tilde{D}_n=\int_{\tilde{U}^{\xi_n}}\prod_{j=1}^n\frac{p(y_j,\delta_j,{\bf x}_j)}{p^*(y_j,\delta_j,{\bf x}_j)}d\pi(\cdot)$ and $\tilde{N}_n=\int_{{{U}^{\epsilon}}^c}\prod_{j=1}^n\frac{p(y_j,\delta_j,{\bf x}_j)}{p^*(y_j,\delta_j,{\bf x}_j)}d\pi(\cdot)$, and $c^*>0$. The last step in the above inequality follows as a result of the following  Lemma \ref{lemthm2}, where we show that  $e^{\epsilon'' M_n^2}\tilde{D}_n$ is greater than one for large $n$ for  any $\epsilon''>0$ with probability one.
 Hence, choosing $\epsilon''>0$ sufficiently small, $E_{p^*}[\Pi_n(U_\epsilon^c|\cdot)]\leq 2e^{-0.5c^*M_n^2}$ for large $n$, and $\Pi_n(U_\epsilon^c|\cdot)<e^{-0.25c^*M_n^2}$ for all but finitely many $n$ almost surely under $p^*$, by Borel-Cantelli Lemma, using similar argument as in Theorem \ref{thm1}. Hence,  $\Pi_n(U_\epsilon^c|\cdot)$ converges to zero almost surely.

\begin{lemma}
Under the set up Theorem \ref{thm2}, $e^{\epsilon'' M_n^2}\tilde{D}_n$ is greater than one for large $n$ with probability one, where $\epsilon''>0$.

\label{lemthm2}
\end{lemma}

\begin{proof}
For Theorem \ref{thm2}, we write \begin{eqnarray}
\prod_{j=1}^n\frac{p(y_j,\delta_j,{\bf x}_j)}{p^*(y_j,\delta_j,{\bf x}_j)}=\large(\prod_{j=1}^n\frac{p(y_j,\delta_j,{\bf x}_j)}{\hat{p}(y_j,\delta_j,{\bf x}_j)}\large)\large(\prod_{j=1}^n\frac{\hat{p}(y_j,\delta_j,{\bf x}_j)}{\hat{p}^*(y_j,\delta_j,{\bf x}_j)}\large)\large(\prod_{j=1}^n\frac{\hat{p}^*(y_j,\delta_j,{\bf x}_j)}{{p}^*(y_j,\delta_j,{\bf x}_j)}\large)=D^{(1)}_n\times D^{(2)}_n\times D^{(3)}_n.
\nonumber
\end{eqnarray}

As in proof of Theorem \ref{thm1}, we will consider $n^{-1}\log D_n^{(1)}, n^{-1}\log D_n^{(2)}, n^{-1}\log D_n^{(3)}$. 

As, $\tilde{g}(\cdot)$ is supported on a compact set, for large $n$, $w^n_{0,i}(t)=\int_0^{T_n}\tilde{g}_i(s)K(l^*_is,l^*_it)ds=\int_0^{\infty}\tilde{g}_i(s)K(l^*_is,l^*_it)ds=w_{0,i}(t)$. Hence, $D_n^{(2)}=1$. 

We choose, $\xi_n=\alpha\frac{M_n^2}{n}$. Choosing, $\alpha$ sufficiently small, we have $|n^{-1}\log D_n^{(1)}|\leq 0.5 \epsilon'' M_n^2/n$, from the argument given in Lemma \ref{dc_prop} (argument for bound  on $n^{-1}\log D_n^{(1)}$ for Theorem \ref{thm1}). Again $-\log \pi(\tilde{U}^{\xi_n})=o(M_n^2)$ (see the following  Lemma \ref{lem2thm2}). 

Using boundedness of $h_{0,i}(\cdot)$'s,  $n^{-1}|\log D_n^{(3)}|\preceq n^{-1}T_n\sum_{j=1}^n {\bf 1}_{{\bf x}_j\in \mathscr{P}^*\Delta \hat{\mathscr{P}^*}}$. As mentioned earlier,  $m_n$ goes to infinity in a polynomial order of $n$ (that is $m_n\sim n^\delta$, for some $\delta>0$),  and therefore for true  tree $T^*$ and partition $\mathscr{P^*}$, we have an approximation with same number of nodes and partition $\hat{\mathscr{P}}^*_n$ such that $d(\hat{\mathscr{P}}^*_n,{\mathscr{P}}^*)\preceq n^{-\delta}$ for some $\delta>0$. 

Now, $n^{-1}E[\sum_{j=1}^n {\bf 1}_{{\bf x}_j\in \mathscr{P}^*\Delta \hat{\mathscr{P}^*}}]\preceq n^{-\delta}$, and hence, from the following Proposition \ref{propthm2},  $|\sum_{j=1}^n {\bf 1}_{{\bf x}_j\in \mathscr{P}^*\Delta \hat{\mathscr{P}^*}}|\preceq n^{1-\delta}+n^{1/2+\delta''}$ for any $\delta''>0$ for large $n$ with probability one. 

Hence, from Proposition \ref{max_sup}, $T_n|\sum_{j=1}^n {\bf 1}_{{\bf x}_j\in \mathscr{P}^*\Delta \hat{\mathscr{P}^*}}|\preceq n^{1-\gamma^*}=o(M_n^2)$, where $0<\gamma^*<min\{\delta, 1/2-\delta''\}$ is a constant, where $0<\delta''<1/2$.

Combining all three the result follows.

\end{proof}
\begin{lemma}
Under the set up of Theorem \ref{thm2}  $-\log \pi(\tilde{U}^{\xi_n})=o(M_n^2)$, where $\tilde{U}^{\xi_n}$ is defined in Theorem \ref{thm2}
 proof before equation \eqref{thm2_eq}. 
 \label{lem2thm2}
\end{lemma}
\begin{proof}
Repeating the construction from Lemma \ref{max_sup}, using $\beta\xi_n$ in place of $c_\xi$, and choosing $\beta>0$ small enough,  $|\hat{w}_{0,i}^n(t)-w_{0,i}(t)|<\frac{\xi_n}{T_n}$. Note that for any $d>0$ there exists $l$ in the support such that $|l-l_i^*|\preceq T_n^{-d}$, from B3 and Proposition \ref{max_sup}. 

Now \[-\log\pi(S^i_{\xi_n,l_i,\tau^2_i})\leq \|\hat{w}_{0,i}^n\|^2_{H_{l_i,n}}-\log(P(\|w^n\|_\infty<\frac{\xi_n}{2T_n}),\]
where both of this term is $o(M_n^2)$,  using the argument  in Lemma \ref{gp_support} for small ball probability, and  as $\xi_n\sim \frac{M_n^2}{n}$, $\|\hat{w}_{0,i}^n\|^2_{H_{l_i,n}}=o(n^\delta)$ and  $\frac{T_n}{\xi_n}=o(n^\delta)$ for any $\delta>0$,  from B3 and Proposition \ref{max_sup}. Hence, this completes the proof.
\end{proof}

\begin{proposition}
Under the set up of Theorem \ref{thm2} and Lemma \ref{lemthm2}, and for  true  tree $T^*$ and partition $\mathscr{P^*}$, and  an approximation with same number of nodes and partition $\hat{\mathscr{P}}^*_n$ such that $d(\hat{\mathscr{P}}^*_n,{\mathscr{P}}^*)\preceq n^{-\delta}$,  we have
$|\sum_{j=1}^n {\bf 1}_{{\bf x}_j\in \mathscr{P}^*\Delta \hat{\mathscr{P}^*}}|\preceq n^{1-\delta}+n^{1/2+\delta''}$ for large $n$ with probability one, for any $0<\delta''<1/2$.
\label{propthm2}
\end{proposition}
\begin{proof}
Suppose, $n^{-1}E[\sum_{j=1}^n {\bf 1}_{{\bf x}_j\in \mathscr{P}^*\Delta \hat{\mathscr{P}^*}}]\leq c_1n^{-\delta}$, $c_1>0$. Then,   using Hoeffding's inequality, \[P\Large(n^{-1}\sum_{j=1}^n {\bf 1}_{{\bf x}_j\in \mathscr{P}^*\Delta \hat{\mathscr{P}^*}}>2c_1n^{-\delta}+n^{-\frac{1}{2}+\delta''}\Large)\preceq e^{-c_2(n^{1-2\delta}+n^{2\delta''})},\]
where $c_2>0$ is a universal constant.  Hence, \[P(n^{-1}\sum_{j=1}^n {\bf 1}_{{\bf x}_j\in \mathscr{P}^*\Delta \hat{\mathscr{P}^*}}>2c_1n^{-\delta}+n^{-\frac{1}{2}+\delta''} \text{infinitely  often })\preceq lim_{m\uparrow \infty}\sum_{n=m}^\infty e^{-c_2(n^{1-2\delta}+n^{2\delta''})}=0,\]
which proves our claim.

\end{proof}

\subsection{Proof of Remark \ref{rmrk1}}

Let $\xi>0$  be any constant. Let, $\hat{g}_{0,i}(t)=\log h_{0,i}(t)-c_{0,i}$ and let ${g}^c_{0,i}(t)$ be its continuous approximation by a piece-wise linear function on $[0,T_n]$ such that $|\hat{g}_{0,i}(t)-{g}^c_{0,i}(t)|=0$ outside a set of Lebesgue measure at most $c_0\frac{c_\xi'}{T_n}$, for small constant $c_\xi'>0$ (to be chosen later) and at the points of differentiability, the absolute value of the first derivative of ${g}^c_{0,i}(t)$  is bounded by $(\frac{c_\xi'}{T_n})^{-1}$. Also, $|{g}^c_{0,i}(t_1)-{g}^c_{0,i}(t_2)|\leq(\frac{c_\xi'}{T_n})^{-1} |t_1-t_2|$.  Here $c_0$ is a constant which is the product of the number of discontinuities and the maximum jump of the step function $\hat{g}_{0,i}(t)$. Construction of this piece-wise linear function is given in the following subsection. Note that $T_n\geq T_{0,i}$ for large $n$ almost surely. 
Without loss of generality, we assume $T_n\geq 1$.

Next, we approximate ${g}^c_{0,i}(t)$ by an element of RKHS. Let, $\phi_l(s)=\frac{1}{2l}K_l(s)$, where $K_l(s-t)=K(s/l,t/l)=e^{-|s-t|/l}$. Let  $l=l_n=\frac{c_\xi}{T^3_n}$.  Let $A_t^{\delta_n}=(t\pm\frac{{c_\xi'}^2}{T_n^2})$. Then, $\int_{|s|\leq \frac{{c_\xi'}^2}{T_n^2}}\phi_l(s)ds=1-o(1/T_n)$  for $l\sim T_n^{-3}$ as $\int_{-\infty}^\infty\phi_1(s)ds =1$. Now, for $\frac{{c_\xi'}^2}{T_n^2}< t<T_n-\frac{{c_\xi'}^2}{T_n^2}$,
\begin{eqnarray*}
| \int_0^{T_n} {g}^c_{0,i}(s)\phi_{l_n}(t-s)ds-{g}^c_{0,i}(t)|&\leq \int_{s\in A_t^{\delta_n}}|g^c_{0,i}(s)-g^c_{0,i}(t)|\phi_{l_n}(t-s)ds+|g^c_{0,i}(t)|\int_{s\in A_t^{\delta_n}}\phi_{l_n}(t-s)ds-1|\\
&+\int_{s\in (A_t^{\delta_n})^c}|g^c_{0,i}(s)|\phi_{l_n}(s)ds\\
&\leq c_0 \frac{T_n}{c_\xi'}\frac{{c_\xi'}^2}{T_n^2}+\frac{\xi}{4T_n}\leq \frac{\xi}{T_n}
\end{eqnarray*}
by first choosing small $c_\xi'$, and then choosing a small $c_\xi$ depending on $c_\xi'$. 

 Let $\hat{{g}}^c_{0,i}(t)$   be a Riemann sum approximation for $\int_0^{T_n} {g}^c_{0,i}(s)\phi_{l_n}(t-s)ds$ with $c_\xi^3T_n^{-d}$ grid size for grids $s_1=0,s_2,\ldots$, and  for $l \in [l_n,l_n+c^3_\xi T_n^{-d}),\tau_i^2\in[1,1+c^3_\xi T_n^{-d}))$. That is, $\hat{{g}}^c_{0,i}(t)=\sum_{j=1}^{\lceil \frac{T_n^{d+1}}{c^3_\xi}\rceil}c_\xi^3T_n^{-d}g^c_{0,i}(s_j)\tau_i^2\phi_l(s_j-t)$. We choose $c_\xi$ small enough such that $|\hat{{g}}^c_{0,i}(t)-\int_0^{T_n} {g}^c_{0,i}(s)\phi_{l_n}(t-s)ds|<\frac{\xi}{T_n}$.   The last step follows, as from some calculation, for a universal constant $k_0>0$, 
\begin{eqnarray*}
 |\hat{{g}}^c_{0,i}(t)-\int_0^{T_n} {g}^c_{0,i}(s)\phi_{l_n}(t-s)ds|\leq \sum_{j=1}^{\lceil \frac{T_n^{d+1}}{c^3_\xi} \rceil}\Large[\int_{s_j}^{s_{j+1}}|g^c_{0,i}(s)-g^c_{0,i}(s_j)|\phi_{l_n}(t-s)ds+&\\\int_{s_j}^{s_{j+1}}|g^c_{0,i}(s_j)|\{|\phi_{l_n}(t-s_j)-\phi_{l_n}(t-s)|+|\phi_{l_n}(t-s_j)-\phi_{l}(t-s_j)|\}ds+& \frac{c_\xi^3}{T_n^{d}}(|\tau_i^2-1|)|g^c_{0,i}(s_j)|\phi_l(s_j-t)]\\
\leq k_0\frac{T_n^3}{c_\xi}\sum_{j=1}^{\lceil\frac{T_n^{d+1}}{c_\xi^3}\rceil}\frac{c_\xi^3}{T_n^{d}}\large[ \frac{c^3_\xi}{c_\xi'T_n^{d-1}}+3\frac{c^2_\xi}{ T_n^{d-3}}\large]\leq 2k_0\frac{c_\xi}{T_n} 
\end{eqnarray*}
for a conservative choice  $d\geq 9$, and a small $c_\xi<c_\xi'$, given the choice of  the small constant $c_\xi'$.  It is easy to check the RKHS norm of $\hat{{g}}^c_{0,i}$ is of the polynomial order of $T_n$ and hence $o(n)$ from Proposition \ref{max_sup}.

Based on the prior  on $l_i$'s as given in Remark \ref{rmrk1}, we have $-\log \pi(l_i\in( [l_n,l_n+c^3_\xi T_n^{-d}))\sim \log n=o(n)$, $-\log \pi(\tau_i^2\in( [1,1+c^3_\xi T_n^{-d}))\sim \log(n)$. From the proof of Lemma \ref{gp_support}, we will have $\rho^2(s,t)\sim T_n^3|s-t|$, for $|s-t|\preceq T_n^{-3}$ and hence, the $\xi/T_n$ covering number for the index set $[0,T_n]$ would be polynomial order of $T_n$ and therefore, from \citet{aurzada2008small}, and Proposition \ref{max_sup},     $-\log P(  sup|\hat{g}^c_{0,i}(\cdot)-w_i(\cdot)|<\xi/2T_n)=o(n)$, for the Gaussian process prior on $w_i(\cdot)$ for $\hat{P}^*_i$.

Next, we select the set  $U^{\epsilon_n}$ in the denominator of equation \eqref{postcal1} in the Proof of Theorem \ref{thm1}. It is induced by the partition  $ \hat{\mathscr{P}}^*$,  such that  $\mu_x(\hat{\mathscr{P}}^*\Delta  {\mathscr{P}}^*)< \xi$, for small $\xi>0$, and  the Gaussian process $w_i$ in $\hat{P}^*_i$ such that  $w_i(\cdot)\in (\hat{{g}}^c_{0,i}(\cdot)\pm \frac{\xi}{T_n})$  in $[0,T_n]$,   and  $\beta_i\in (c_{0,i}\pm \frac{\tilde{c}_\xi}{T_n})$  for the small constant $\tilde{c}_\xi$. Without loss of generality we can assume it to be the constant $c_\xi$ needed for the approximation earlier. For the set $U^{\epsilon_n}$, the negative log prior probability is of $o(n)$, as seen in Theorem \ref{thm1}.

Now from \ref{postcal1}, 
\begin{eqnarray}
\prod_{j=1}^n\frac{p(y_j,\delta_j,{\bf x}_j)}{p^*(y_j,\delta_j,{\bf x}_j)}=\large(\prod_{j=1}^n\frac{p(y_j,\delta_j,{\bf x}_j)}{\hat{p}(y_j,\delta_j,{\bf x}_j)}\large)\large(\prod_{j=1}^n\frac{\hat{p}(y_j,\delta_j,{\bf x}_j)}{\hat{p}^*(y_j,\delta_j,{\bf x}_j)}\large)\large(\prod_{j=1}^n\frac{\hat{p}^*(y_j,\delta_j,{\bf x}_j)}{{p}^*(y_j,\delta_j,{\bf x}_j)}\large)=D^{(1)}_n\times D^{(2)}_n\times D^{(3)}_n
\label{decomp2}
\end{eqnarray}
where, $p^*$, $\hat{p}^*$ is defined the same as in equation \ref{decomp}, $p$ is a density in $U^{\epsilon_n}$ and $\hat{p}$ is the density using $\beta_i=c_{0,i}$ and  $\hat{g}^c_{0,i}(t)$ $[0,T_n]$ instead of $\hat{g}_{0,i}(t)$'s in $\hat{P}^*_i$'s. 

The calculation for $D_n^{(1)}, D_n^{(3)}$ is therefore the same as the Theorem \ref{thm1} proof. For calculation $D_n^{(2)}$, note that if $\hat{h}^c_{0,i}(t), H^c_{0,i}$ corresponds to  $\hat{g}^c_{0,i}(t)$ then, $|\hat{H}^c_{0,i}-H_{0,i}|\preceq \xi$ in $[0,T_n]$ from the construction of ${g}^c_{0,i}(t)$. Also,  $\hat{h}^c_{0,i}(t)=\hat{h}^c_{0,i}(T_n)$ for $t>T_n$.

Hence, 
\begin{eqnarray}
n^{-1}|\log D_n^{(2)}| \leq n^{-1}\sum_{j=1}^n\sum_i|\hat{H}^c_{0,i}(y_j)-H_{0,i}(y_j)||{\bf I}_{{\bf x_j} \in \hat{P}^*_i}+n^{-1}\sum_{j=1}^n \sum_i|\log  \frac{\hat{h}^c_{0,i}(y_j)}{h_{0,i}(y_j)}|{\bf I}_{{\bf x_j} \in \hat{P}^*_i}\preceq \xi
\label{dn2bd}
\end{eqnarray}
for large $n$, almost surely.  The above  result  holds by the construction of  $\hat{g}^c_{0,i}(t)$ which for small $c_\xi,c_\xi'$, will remain within a $\frac{2\xi}{T_n}$ supremum neighborhood  of $ \hat{g}_{0,i}(t)$ outside a set of Lebesgue measure less than $\frac{\xi}{T_n}$ in $[0,T_n]$ and $\log  \hat{h}^c_{0,i}, \log {h}_{0,i}$ are bounded, and from the fact that the proportion of $y_j$'s that lie within 
that set where $\hat{g}^c_{0,i}(t), \hat{g}_{0,i}(t)$ are  more than $\frac{4\xi}{T_n}$ apart will be at most of the order of $\xi$ almost surely, for large $n$.

The last step follows as,   if  the measure of that set is less than $\alpha\xi$ under $p^*$ for all ${\bf x}$ (where $\alpha>0$),   then for ${\bf x}_j$'s in  $P^*_i$, the probability that the proportion of $y_j$'s in that smaller than $\alpha\xi$ measure set  is greater than $2\alpha\xi$,  would be bounded above by $e^{-2n_i\alpha^2\xi^2}$  by the Hoeffding inequality, where there are $n_i$ many observations in ${P}^*_i$. By the law of large numbers $n_i>n(1-\alpha')\mu_x({P}^*_i)$ for large $n$, almost surely, for any $0<\alpha'<1$.  Then, from the Borel-Cantelli lemma, we have the result on the almost sure  bound on the proportion for large $n$. Hence, equation \eqref{dn2bd} is satisfied and the proof of Theorem \ref{thm1} follows by selecting $\xi$ small.

\subsubsection{Construction for Remark \ref{rmrk1}}
A simple construction for the continuous function for Remark \ref{rmrk1} can be written as follows. Let $\hat{g}_{0,i}(t)=c_{1,i}$ for $t<t^*_{1,i}$, and  $\hat{g}_{0,i}(t)=c_{2,i}$ for $t^*_{1,i}\leq t <t^*_{2,i}$, and so on, and $\hat{g}_{0,i}(t)=c_{d_i,i}=0$ for $t\geq t^*_{d_i,i}=T_{0,i}$. 

Let $t^*_{0,i}=0$. We define $g^c_{0,i}(t)=c_{k,i}$ for $t\in[t^*_{k,i},t^*_{k+1,i}-\frac{c_\xi}{T_n}|c_{k,i}-c_{k+1,i}|]$, for $k=0,\dots,d_i-1$, and it is linear between, $t=t^*_{k+1,i}-\frac{c_\xi}{T_n}|c_{k,i}-c_{k+1,i}|$ and $t^*_{k+1,i}$ for $k=1,\dots, d_i-1$. Also, note that  $t^*_{k,i}<t^*_{k+1,i}-\frac{c_\xi}{T_n}|c_{k,i}-c_{k+1,i}|$  for $c_\xi$ small
and the intervals are well defined.

\subsection{Proof of Remark \ref{rmrk2}}
We use $\log h_{0,i}(t)$ in place of the continuous approximation ${g}^c_{0,i}(t)$  in Remark 1, between $0$ and $T_n$. 

Next of the calculation follows similarly for prior probability, with  $U^{\epsilon_n}$ in the denominator in the proof of Theorem \ref{thm1} is induced by the approximating partition $ \hat{\mathscr{P}}^*$,  such that  $d(\hat{\mathscr{P}}^*,{\mathscr{P}}^*)< \xi$, for small $\xi>0$, and   $w_i(\cdot)\in (\hat{{g}}^c_{0,i}\pm \frac{\xi}{T_n})$ in $\hat{P}^*_i$   and  $\beta_i\in (0\pm \frac{c_\xi}{T_n})$  for a small constant $c_\xi$. For this set,  the negative log-prior probability is of $o(n)$ similar to Remark \ref{rmrk1}. 

Next, for computing the denominator in equation \ref{postcal1}, we use  \begin{eqnarray*}
\prod_{j=1}^n\frac{p(y_j,\delta_j,{\bf x}_j)}{p^*(y_j,\delta_j,{\bf x}_j)}=\large(\prod_{j=1}^n\frac{p(y_j,\delta_j,{\bf x}_j)}{\hat{p}(y_j,\delta_j,{\bf x}_j)}\large)\large(\prod_{j=1}^n\frac{\hat{p}(y_j,\delta_j,{\bf x}_j)}{\hat{p}^*(y_j,\delta_j,{\bf x}_j)}\large)\large(\prod_{j=1}^n\frac{\hat{p}^*(y_j,\delta_j,{\bf x}_j)}{{p}^*(y_j,\delta_j,{\bf x}_j)}\large)=D^{(1)}_n\times D^{(2)}_n\times D^{(3)}_n
\label{decomp3}
\end{eqnarray*}
where $p^*$, $\hat{p}^*$ is defined the same as in equation \eqref{decomp2}. Here, $\hat{p}$ is defined with hazard function $\hat{h}^c_{0,i}$. Here, $\log \hat{h}^c_{0,i}(t)=\hat{{g}}^c_{0,i}$ is the Riemann approximation with grid size $c_\xi^3 T_n^{-d}$ for the convolution representation for $\log \hat{h}_{0,i}(t)$  on $[0,T_n]$  with kernel $\phi_l(\cdot)$ with $l \in [l_n,l_n+c^3_\xi T_n^{-d})$, as in Remark \ref{rmrk1}, where   $\hat{h}_{0,i}(t)={h}_{0,i}(t)$ for $0\leq t\leq T_n$ and $\hat{h}_{0,i}(t)=h_{0,i}(T_n)$ for $t>T_n$.  Therefore, the argument for $n^{-1}\log D_n^{(2)}$ from the proof of Remark \ref{rmrk1} can be applied, as  by construction,    for sufficiently small $c_\xi,c_\xi'$, $\log \hat{h}^c_{0,i}(t)$ will remain within $\frac{\xi}{T_n}$ supremum neighborhood  of $\log \hat{h}_{0,i}(t)$ outside a set of Lebesgue measure less than $\frac{\xi}{T_n}$ in $[0,T_n]$.

\subsection{Proof of Remark \ref{rmrk3}}
We write equation \eqref{decomp} as\begin{eqnarray*}
\prod_{j=1}^n\frac{p(y_j,\delta_j,{\bf x}_j)}{p^*(y_j,\delta_j,{\bf x}_j)}=\large(\prod_{j=1}^n\frac{p(y_j,\delta_j,{\bf x}_j)}{\hat{p}(y_j,\delta_j,{\bf x}_j)}\large)\large(\prod_{j=1}^n\frac{\hat{p}(y_j,\delta_j,{\bf x}_j)}{\hat{p}^*(y_j,\delta_j,{\bf x}_j)}\large)\large(\prod_{j=1}^n\frac{\hat{p}^*(y_j,\delta_j,{\bf x}_j)}{\hat{p}_1^*(y_j,\delta_j,{\bf x}_j)}\large)\large(\prod_{j=1}^n\frac{\hat{p}_1^*(y_j,\delta_j,{\bf x}_j)}{{p}^*(y_j,\delta_j,{\bf x}_j)}\large)
\end{eqnarray*}
where, $p^*(\cdot)$ is same as in equation \eqref{decomp}. Let $T^*_1$ be a tree for partition $\mathscr{P}^*_1$, within $\epsilon'$  distance of the true partition $\mathscr{P}^*$, and the tree $\hat{T}^*$ (and corresponding $\hat{\mathscr{P}}^*$) be its approximation such that $d(\hat{\mathscr{P}}^*,\mathscr{P}^*_1)<\epsilon'$, and $-\log\pi(\hat{T}^*)=o(n)$. We use true hazard functions for the true partitions on the approximating partitions for  $\hat{\mathscr{P}}^*$ and $\mathscr{P}^*_1$, for $\hat{p}^*(\cdot)$ and $\hat{p}^*_1(\cdot)$, respectively. Then,  $\hat{p}(\cdot)$, $p(\cdot)$ can be defined as in equation \eqref{decomp}. 

The calculations for $\prod_{j=1}^n\frac{\hat{p}^*(y_j,\delta_j,{\bf x}_j)}{\hat{p}_1^*(y_j,\delta_j,{\bf x}_j)}$ and  $\prod_{j=1}^n\frac{\hat{p}_1^*(y_j,\delta_j,{\bf x}_j)}{{p}^*(y_j,\delta_j,{\bf x}_j)}$ can be addressed in the same way as  for $D_n^{(3)}$ in Theorem \ref{thm1} and in Lemma \ref{dc_prop}. Hence, the proof follows.

\end{document}